%% file: LargeNarrowProofs.tex
\documentclass[11pt,twoside,a4paper]{article}

\usepackage{ifthen}

\newboolean{conferenceversion}
\setboolean{conferenceversion}{false}

\usepackage[a4paper,margin=2.5cm]{geometry}

\usepackage[english]{babel}

\usepackage{fancyhdr}

\usepackage[sf,SF]{subfigure}

\newcommand{\email}[1]{{\tt #1}}

\usepackage{bussproofs}

\usepackage{enumitem}

\input{preamble.tex}
\input{notationProofComplexity.tex}

\input{notationLocal.tex}
\input{editingmacros.tex}
\input{albertmacros.tex}

\numberwithin{equation}{section}

\begin{document}

%
%

\title{\textbf{Narrow Proofs May Be Maximally Long}%
  \thanks{This is the full-length version of the paper~%
    \cite{ALN14NarrowProofs},
    which appeared in \emph{Proceedings of the 29th Annual {IEEE} 
      Conference on Computational Complexity ({CCC}~'14)}.}}

\author{%
  Albert Atserias \\
  Universitat Polit\`ecnica de Catalunya\\
  \email{atserias@cs.upc.edu}
  \and
  Massimo Lauria \\
  KTH Royal Institute of Technology  \\
  \email{lauria@kth.se}
  \and
  Jakob Nordstr\"om \\
  KTH Royal Institute of Technology  \\
  \email{jakobn@kth.se}
}

\date{\today}

\maketitle

%
%



%
%

\input{abstract.tex}

%
%

\thispagestyle{empty}

%

\pagestyle{fancy}     
\fancyhead{}
\fancyfoot{}
\renewcommand{\headrulewidth}{0pt}
\renewcommand{\footrulewidth}{0pt}

%
%
\fancyhead[CE]{\slshape NARROW PROOFS MAY BE MAXIMALLY LONG}
\fancyhead[CO]{\slshape \nouppercase{\leftmark}\/}
\fancyfoot[C]{\thepage}

\setlength{\headheight}{13.6pt}

%
%

\input{introduction.tex}

\input{ourresults.tex}

\input{paperoutline.tex}

\input{preliminaries.tex}

\input{narrowproofs.tex}

\input{algebraicandsemialgebraic.tex}

\input{lasserreupperbound.tex}

\input{concludingremarks.tex}

\input{acknowledgements.tex}

%
%

\bibliography{refArticlesUTF8,refBooksUTF8,refOtherUTF8,refLocalUTF8}
\bibliographystyle{alpha}

\end{document}

%% file: preamble.tex
\usepackage{ifthen}

\input{doctypedetection.tex}

\usepackage[utf8]{inputenc}
\usepackage{amsmath}
\usepackage{amssymb}
\usepackage{amsfonts}

\usepackage{mathtools}

\ifthenelse
{\boolean{maybeElsevier}}
{}
{\usepackage{varioref}}

\usepackage{xspace}

\ifthenelse    
{\boolean{maybeSTOC} \or \boolean{maybeSIAM} \or \boolean{maybeLMCS}
 \or \boolean{maybeNOW} \or \boolean{maybeLNCS} \or \boolean{maybeACM}}
{}
{\usepackage{amsthm}
 \usepackage{amsthm-boldfont}}

\ifthenelse
{\boolean{maybeElsevier}    \or \boolean{maybeFOCS} \or \boolean{maybeSTOC}
  \or \boolean{maybePoster} \or \boolean{maybeSIAM} \or \boolean{maybeLMCS}
  \or \boolean{maybeIEEE}   \or \boolean{maybeNOW}  \or \boolean{maybeICS}
  \or \boolean{maybeThesis} \or \boolean{maybeLNCS} \or \boolean{maybeACM}} 
{}
{\usepackage{nada-typography}}

\input{definitions.tex}

\ifthenelse{\isundefined{\DONOTLOADHYPERREFPACKAGE}
  \and \not \boolean{maybeSTOC}     \and \not \boolean{maybeFOCS}
  \and \not \boolean{maybeElsevier} \and \not \boolean{maybePoster} 
  \and \not \boolean{maybeSIAM}     \and \not \boolean{maybeACM}
  \and \not \boolean{maybeIEEE}     \and \not \boolean{maybeNOW}
  \and \not \boolean{maybeICS}      \and \not \boolean{maybeThesis}
  \and \not \boolean{maybeLNCS}}
{\usepackage[bookmarks=true,
        bookmarksnumbered=true,
        pagebackref=true,
        colorlinks=true,
        linkcolor=blue,
        citecolor=blue]{hyperref}}
{}

\ifthenelse{\boolean{maybeSTOC}}                  
  {\input{definitions-STOC.tex}}
  {\ifthenelse{\boolean{maybeSIAM}}
    {\input{definitions-SIAM.tex}}
    {\ifthenelse{\boolean{maybeNOW}}
      {\input{definitions-NOW.tex}}
      {\ifthenelse{\boolean{maybeLMCS}}
        {\input{definitions-lmcs.tex}}
        {\ifthenelse{\boolean{maybeIEEE}}
          {\input{definitions-IEEE.tex}}
          {\ifthenelse{\boolean{maybeICS}}
            {\input{definitions-ICS.tex}}
            {\ifthenelse{\boolean{maybeLNCS}}
              {\input{definitions-LNCS.tex}}
              {\ifthenelse{\boolean{maybeACM}}
                {\input{definitions-ACM.tex}}
                {\ifthenelse{\boolean{maybeArticle} \or \boolean{maybeElsevier}}
                  {\input{definitions-article.tex}}
                  {\input{definitions-report.tex}}}}}}}}}}

\ifthenelse{\boolean{maybeThesis}}
{}
{\usepackage{ifpdf}}

\usepackage{graphicx}  

\ifpdf         
\fi

\ifthenelse
{\boolean{maybeSTOC} \or \boolean{maybeThesis} \or \boolean{maybeLNCS}}
{}
{
  \setcounter{secnumdepth}{3}
  \setcounter{tocdepth}{3}
}

\newcount\hours
\newcount\minutes
\def\SetTime{\hours=\time
\global\divide\hours by 60
\minutes=\hours
\multiply\minutes by 60
\advance\minutes by-\time
\global\multiply\minutes by-1 }
\SetTime
\def\now{\number\hours:\ifnum\minutes<10 0\fi\number\minutes}

%% file: doctypedetection.tex
\usepackage{ifthen}

\newboolean{maybeSTOC}
\newboolean{maybeFOCS}
\newboolean{maybeElsevier}
\newboolean{maybeNOW}
\newboolean{maybeLMCS}
\newboolean{maybePoster}
\newboolean{maybeSIAM}
\newboolean{maybeIEEE}
\newboolean{maybeICS}
\newboolean{maybeLNCS}
\newboolean{maybeACM}
\newboolean{maybeArticle}
\newboolean{maybeReport}
\newboolean{maybeThesis}

\ifthenelse{\isundefined{\affaddr}}        
{\setboolean{maybeSTOC}{false}}
{\setboolean{maybeSTOC}{true}}

\ifthenelse{\not \isundefined{\affiliation}
  \and \not \isundefined{\Section}}        
{\setboolean{maybeFOCS}{true}}
{\setboolean{maybeFOCS}{false}}

\ifthenelse{\not \isundefined{\tnotetext}
  \and \not \isundefined{\ead}}        
{\setboolean{maybeElsevier}{true}}
{\setboolean{maybeElsevier}{false}}

\ifthenelse{\not \isundefined{\journal} \and 
  \not \isundefined{\volume} \and 
  \not \isundefined{\issue} \and 
  \not \isundefined{\copyrightowner} \and 
  \not \isundefined{\articletitle} \and 
  \not \isundefined{\pubyear}}
{\setboolean{maybeNOW}{true}}
{\setboolean{maybeNOW}{false}}

\ifthenelse{\not \isundefined{\keywords} \and
  \not \isundefined{\subjclass} \and
  \not \isundefined{\titlecomment} \and
  \not \isundefined{\revisionname} \and
  \not \isundefined{\lmcsheading}}
{\setboolean{maybeLMCS}{true}}
{\setboolean{maybeLMCS}{false}}

\ifthenelse{\not \isundefined{\IEEEtransversionmajor} \and 
  \not \isundefined{\IEEEtransversionminor} \and 
  \not \isundefined{\IEEEitemize} \and 
  \not \isundefined{\IEEEenumerate} \and 
  \not \isundefined{\IEEEdescription}}
{\setboolean{maybeIEEE}{true}}
{\setboolean{maybeIEEE}{false}}

\ifthenelse{\isundefined{\sicomp}}
{\setboolean{maybeSIAM}{false}}
{\setboolean{maybeSIAM}{true}}

\ifthenelse{\not \isundefined{\address} \and 
  \not \isundefined{\email} \and 
  \not \isundefined{\keywords} \and 
  \not \isundefined{\sanhao} \and 
  \not \isundefined{\wuhao}}
{\setboolean{maybeICS}{true}}
{\setboolean{maybeICS}{false}}

\ifthenelse{\not \isundefined{\institute} \and 
  \not \isundefined{\institutename} \and 
  \not \isundefined{\email} \and 
  \not \isundefined{\fnmsep}}
{\setboolean{maybeLNCS}{true}}
{\setboolean{maybeLNCS}{false}}

\ifthenelse{\not \isundefined{\acmVolume} \and 
  \not \isundefined{\acmNumber} \and 
  \not \isundefined{\acmArticle} \and 
  \not \isundefined{\acmYear} \and 
  \not \isundefined{\acmMonth}}
{\setboolean{maybeACM}{true}}
{\setboolean{maybeACM}{false}}

\ifthenelse{\isundefined{\conference}}
{\setboolean{maybePoster}{false}}
{\setboolean{maybePoster}{true}}

\ifthenelse{\isundefined{\chapter}}
{\setboolean{maybeArticle}{true}}
{\setboolean{maybeArticle}{false}}

\ifthenelse{\not \isundefined{\examen} 
  \and \not \isundefined{\disputationsdatum} 
  \and \not \isundefined{\disputationslokal}}   
  {\setboolean{maybeThesis}{true}}
  {\setboolean{maybeThesis}{false}}
           
\ifthenelse{\boolean{maybeArticle} \or \boolean{maybeThesis}
  \or \boolean{maybeSTOC} \or \boolean{maybeFOCS}
  \or \boolean{maybeSIAM} \or \boolean{maybeIEEE}
  \or \boolean{maybeICS} \or \boolean{maybePoster}}
{\setboolean{maybeReport}{false}}
{\setboolean{maybeReport}{true}}

%% file: definitions.tex
\usepackage{ifthen}
\usepackage{xspace}
\ifthenelse
{\boolean{maybeElsevier}}
{}
{\usepackage{varioref}}

\DeclareMathAlphabet{\mathsfsl}{OT1}{cmss}{m}{sl}

\newcommand{\eqperiod}{\enspace .}
\newcommand{\eqcomma}{\enspace ,}

\newcommand{\introduceterm}[1]{{\emph{#1}}}

\newcommand{\ie}{i.e.,\ }

\newcommand{\etal}{et al.\@\xspace}

\ifthenelse{\boolean{maybeIEEE}}{}{}

\newcommand{\Bigoh}[1]{\mathrm{O} \bigl( #1 \bigr)}
\newcommand{\bigoh}[1]{\mathrm{O} ( #1 )}

\newcommand{\littleoh}[1]{\mathrm{o} ( #1 )}

\newcommand{\Bigomega}[1]{\Omega \bigl( #1 \bigr)}
\newcommand{\bigomega}[1]{\Omega ( #1 )}

\newcommand{\problemlanguageformat}[1]{\textsc{#1}\xspace}

\newcommand{\MINIMALUNSATISFIABILITY}%
  {\problemlanguageformat{minimal unsatisfiability}}

\newcommand{\complclassformat}[1]{\textrm{\upshape{\textsf{#1}}}\xspace}

\newcommand{\cocomplclass}[1]%
        {\mbox{\complclassformat{co}-\complclassformat{#1}}\xspace}

\newcommand{\Pclass}{\complclassformat{P}}
\newcommand{\NP}{\complclassformat{NP}}
\newcommand{\NPclass}{\NP}
\newcommand{\coNP}{\cocomplclass{NP}}

\newcommand{\EXPTIME}{\complclassformat{EXPTIME}}

\newcommand{\refsec}[1]{Section~\ref{#1}}

\newcommand{\refth}[1]{Theorem~\ref{#1}}

\newcommand{\reflem}[1]{Lemma~\ref{#1}}
\newcommand{\reftwolems}[2]{Lemmas~\ref{#1} and~\ref{#2}}

\newcommand{\refrem}[1]{Remark~\ref{#1}}

\newcommand{\Refth}[1]{Theorem~\ref{#1}}

\ifthenelse
{\isundefined{\refeq}}
{\newcommand{\refeq}[1]{\eqref{#1}}}
{\renewcommand{\refeq}[1]{\eqref{#1}}}

\newcommand{\ceiling}[1]{\lceil #1 \rceil}

\newcommand{\Rplus}     {\mathbb{R}^{+}}

\newcommand{\F}{\mathbb{F}}

\DeclareMathOperator{\Expop}{E}

\newcommand{\Prob}[2][]{\Pr_{#1} \bigl[ #2 \bigr]}
\ifthenelse{\boolean{maybeLMCS}}
{}
{}

\newcommand{\twincommandJN}[6]%
    {#1#2#3\vphantom{#2#5}\mspace{-2.25mu}#4.#5#6}

\newcommand{\CondExp}[2]%
    {\Expop\twincommandJN{\bigl[}{#1}{\bigl|}{\bigr}{\,#2}{\bigr]}}
\newcommand{\CONDEXP}[2]%
     {\Expop\twincommandJN{\left[}{#1}{\left|}{\right}{\,#2}{\right]}}

\newcommand{\CondProb}[3][]%
    {\Pr_{#1}\twincommandJN{\bigl[}{#2}{\bigl|}{\bigr}{\,#3}{\bigr]}}
\newcommand{\CONDPROB}[3][]%
    {\Pr_{#1}\twincommandJN{\left[}{#2}{\left|}{\right}{\,#3}{\right]}}

\newcommand{\isdistras}[2]{\ensuremath{#1} \sim \ensuremath{#2}}

\newcommand{\funcdescr}[3]{\ensuremath{ #1 : #2 \to #3}}

\newcommand{\setcompact}[1]{{\ensuremath{\bigl\{ #1 \bigr\}}}}

\newcommand{\setdescrcompact}[3][\mid]{{\setcompact{ #2 #1 #3 }}}

\newcommand{\set}[1]{\{ #1 \}}

\newcommand{\Setdescr}[3][|]%
     {\twincommandJN{\bigl\{}{#2}{\bigl#1}{\bigr}{\,#3}{\bigr\}}}
\newcommand{\SETDESCR}[3][|]%
     {\twincommandJN{\left\{}{#2}{\left#1}{\right}{\,#3}{\right\}}}

\newcommand{\Setdescrbrackets}[3][|]%
     {\twincommandJN{\bigl[}{#2}{\bigl#1}{\bigr}{\,#3}{\bigr]}}
\newcommand{\SETDESCRBRACKETS}[3][|]%
     {\twincommandJN{\left[}{#2}{\left#1}{\right}{\,#3}{\right]}}

\newcommand{\setsize}[1]{\lvert#1\rvert}

\newcommand{\intersection}{\cap}

\newcommand{\DisjointunionInText}%
    {{\smash{\overset{\mbox{\boldmath{.}}}{\bigcup}}}\vphantom{\bigcup}}

\newcommand{\intnfirst}[1]{[{#1}]}

\newcommand{\Lor}{\bigvee}

\newcommand{\olnot}[1]{\overline{#1}}
\newcommand{\stdnot}[1]{\olnot{#1}}
\newcommand{\cnfform}{\cnfshort for\-mu\-la\xspace}

\newcommand{\cnfshort}{CNF\xspace}

\newcommand{\xcnfform}[1]{\mbox{\ensuremath{#1}-}\cnfform}
\newcommand{\kcnfform}{\xcnfform{\clwidth}}
\newcommand{\xclause}[1]{\mbox{\ensuremath{#1}-clause}\xspace}

\newcommand{\nvar}{n}
\newcommand{\nclause}{m}
\newcommand{\clwidth}{k}

\newcommand{\randkcnfnclwrepl}[3][\clwidth]%
        {\ensuremath{\mathcal{F}^{#2, #3}_{#1}}}
\newcommand{\randkcnfnclwreplstd}%
        {\randkcnfnclwrepl{\clwidth}{\nvar}{\nclause}}

\newcommand{\israndkcnfnclwrepl}[4]%
  {\isdistras{#1}{\randkcnfnclwrepl[#2]{#3}{#4}}}

\newcommand{\randkcnfprobcl}[3]%
        {\ensuremath{\mathcal{F}^{#2}_{#1} \bigl(#3 \bigr)}}

\newcommand{\pcfor}[4][to]{for #2 := #3 #1 #4 do}
\newcommand{\pcformath}[4][to]%
    {\pcfor[#1]{\ensuremath{#2}}{\ensuremath{#3}}{\ensuremath{#4}}}

\newcommand{\pcassigncompact}[2]{#1 := #2}
\newcommand{\pcassignmathcompact}[2]%
        {\pcassigncompact{\ensuremath{#1}}{\ensuremath{#2}}}

\newcommand{\inductionformat}[1]{\textit{#1}}

\newcommand{\BASE}[1][]
        {\inductionformat
                {%
                        \ifthenelse{\equal{#1}{}}%
                                {Base case: }%
                                {Base case (#1):}%
                }%
        }



%% file: definitions-IEEE.tex
%
%

%
%
%

%
%
%

%
%
%
%

%
%

\newtheorem{standardlocalcounter}{Dummy}[section]
\newtheorem{standardglobalcounter}{Dummy}

\theoremstyle{plain}    

\newtheorem{theorem}[standardglobalcounter]{Theorem}
\newtheorem{lemma}[standardglobalcounter]{Lemma}
\newtheorem{proposition}[standardglobalcounter]{Proposition}
\newtheorem{corollary}[standardglobalcounter]{Corollary}
\newtheorem{observation}[standardglobalcounter]{Observation}
\newtheorem{fact}[standardglobalcounter]{Fact}

\newtheorem{conjecture}[standardglobalcounter]{Conjecture}
\newtheorem{openquestion}{Open Question}
\newtheorem{openproblem}{Open Problem}
\newtheorem{problem}{Problem}

\theoremstyle{definition}

\newtheorem{property}[standardglobalcounter]{Property}
\newtheorem{definition}[standardglobalcounter]{Definition}
\newtheorem{claim}[standardglobalcounter]{Claim}

\theoremstyle{remark}
\newtheorem{remark}[standardglobalcounter]{Remark}
\newtheorem{example}[standardglobalcounter]{Example}

%
%

\newtheoremstyle{meta}
  {3pt}
  {3pt}
  {\scshape \small }
  {}
  {\scshape \small }
  {:}
  { }
  {}

\theoremstyle{meta}
\newtheorem{meta}{Meta comment}

\newtheoremstyle{questions}
  {3pt}
  {3pt}
  {\sffamily \slshape}
  {}
  {\bfseries \sffamily \slshape}
  {:}
  { }
  {}

\theoremstyle{questions}
\newtheorem{questions}{Open questions}

%% file: definitions-LNCS.tex
%
%
%

%
%

%
%
%

%
%
%
%

\spnewtheorem*{proofsketch}{Proof sketch}{\itshape}{\rmfamily}

\spnewtheorem{observation}{Observation}{\bfseries}{\itshape}
\spnewtheorem{fact}{Fact}{\bfseries}{\itshape}

%
%
%

%% file: definitions-article.tex
%
%
%

%
%

%
%
%

%
%
%
%

%
%

\newtheorem{standardlocalcounter}{Dummy}[section]
\newtheorem{standardglobalcounter}{Dummy}

\theoremstyle{plain}    

\newtheorem{theorem}[standardlocalcounter]{Theorem}
\newtheorem{lemma}[standardlocalcounter]{Lemma}

\newtheorem{corollary}[standardlocalcounter]{Corollary}

\theoremstyle{definition}

\theoremstyle{remark}
\newtheorem{remark}[standardlocalcounter]{Remark}

%
%

\newtheoremstyle{meta}
  {3pt}
  {3pt}
  {\scshape \small }
  {}
  {\scshape \small }
  {:}
  { }
  {}

\theoremstyle{meta}

\newtheoremstyle{questions}
  {3pt}
  {3pt}
  {\sffamily \slshape}
  {}
  {\bfseries \sffamily \slshape}
  {:}
  { }
  {}

\theoremstyle{questions}

%% file: notationProofComplexity.tex
\newcommand{\formuladots}{\cdots}

\newcommand{\proofstd}{\ensuremath{\pi}}

\newcommand{\deriveswithall}%
        {\vdash_{\!\!\!{\scriptscriptstyle \forall}}} 
\newcommand{\notderiveswithall}%
        {\nvdash_{\!\!\!{\scriptscriptstyle \forall}}} 

\newcommand{\clcfgtransitioncrammed}[2]%
        {\ensuremath{#1 \!\rightsquigarrow\! #2}}

\newcommand{\fstd}{{\ensuremath{F}}}

\newcommand{\emptycl}{0}
\renewcommand{\emptycl}{\bot}

\newcommand{\varx}{\ensuremath{x}}

\newcommand{\lita}{\ensuremath{a}}

\newcommand{\cla}{\ensuremath{A}}
\newcommand{\clb}{\ensuremath{B}}
\newcommand{\clc}{\ensuremath{C}}

\newcommand{\setsofvarsorlitlarge}[2]%
        {\mathit{#1}\left({#2}\right)}
\newcommand{\setsofvarsorlit}[2]%
        {\mathit{#1}({#2})}
\newcommand{\setsofvarsorlitcompact}[2]%
        {\mathit{#1}\bigl({#2}\bigr)}

\newcommand{\setsofvarsorlitsup}[3]%
        {\mathit{#1}^{#2}({#3})}
\newcommand{\setsofvarsorlitsuplarge}[3]%
        {\mathit{#1}^{#2}\left({#3}\right)}
\newcommand{\setsofvarsorlitsupcompact}[3]%
        {\mathit{#1}^{#2}\bigl({#3}\bigr)}

\newcommand{\restr}{\ensuremath{\rho}}
\newcommand{\rstd}{\restr}

\newcommand{\restrict}[2]{{\ensuremath{{#1\!\!\upharpoonright}_{#2}}}}

\newcommand{\derivabbrev}[2]{\bigl( #1 \vdash #2 \bigr)}
\newcommand{\derivabbrevsmall}[2]{( #1 \vdash #2 )}
\newcommand{\derivabbrevcompact}[2]{\bigl( #1 \vdash #2 \bigr)}

\newcommand{\refutabbrevsmall}[1]{\derivabbrevsmall{#1}{\falsenum}}
\newcommand{\refutabbrevcompact}[1]{\derivabbrevcompact{#1}{\falsenum}}

\renewcommand{\refutabbrevsmall}[1]{\derivabbrevsmall{#1}{\emptycl}}
\renewcommand{\refutabbrevcompact}[1]{\derivabbrevcompact{#1}{\emptycl}}

\newcommand{\genericmeasure}[2]{{\mathit{#1}}_{#2}}

\newcommand{\genericrefsmall}[3]%
    {{\mathit{#1}}_{#2}\refutabbrevsmall{#3}}
\newcommand{\genericrefcompact}[3]%
    {{\mathit{#1}}_{#2}\refutabbrevcompact{#3}}
\newcommand{\genericderiv}[4]%
    {{\mathit{#1}}_{#2}\derivabbrev{#3}{#4}}
\newcommand{\genericderivsmall}[4]%
    {{\mathit{#1}}_{#2}\derivabbrevsmall{#3}{#4}}
\newcommand{\genericderivcompact}[4]%
    {{\mathit{#1}}_{#2}\derivabbrevcompact{#3}{#4}}
\newcommand{\generictaut}[3]%
    {{\mathit{#1}}_{#2}\derivabbrev{}{#3}}
\newcommand{\generictautcompact}[3]%
    {{\mathit{#1}}_{#2}\derivabbrevcompact{}{#3}}
\newcommand{\generictautsmall}[3]%
    {{\mathit{#1}}_{#2}\derivabbrevsmall{}{#3}}
\newcommand{\length}[1][]{\genericmeasure{L}{#1}}

\newcommand{\lengthstd}{\length}

\newcommand{\formulaformat}[1]{\ensuremath{\mathit{#1}}}
\renewcommand{\formulaformat}[1]{\mathit{#1}}

\newcommand{\transitionarrow}{\rightsquigarrow}
\newcommand{\pebcfgtransition}[2]%
    {\ensuremath{#1 \transitionarrow #2}}
\newcommand{\pebcfgtransitionsqueeze}[2]%
    {#1 \! \transitionarrow \! #2}

\newcommand{\formatpebblingprice}[1]{\text{\textsl{\textsf{#1}}}}

\newcommand{\Pebblingprice}[1]%
    {\formatpebblingprice{Peb}\bigl(#1\bigr)}
\newcommand{\pebblingpricecompact}[1]
    {\formatpebblingprice{Peb}\bigl(#1\bigr)}

\newcommand{\Bwpebblingprice}[1]%
    {\formatpebblingprice{BW-Peb}\bigl(#1\bigr)}
\newcommand{\bwpebblingpricecompact}[1]
    {\formatpebblingprice{BW-Peb}\bigl(#1\bigr)}

\newcommand{\pebpersistentsymbol}{\bullet}
\newcommand{\pebvisitingsymbol}{\emptyset}

\newcommand{\bwpebpricepersistent}[1]%
    {\formatpebblingprice{BW-Peb}^{\pebpersistentsymbol}(#1)}
\newcommand{\Bwpebpricepersistent}[1]%
    {\formatpebblingprice{BW-Peb}^{\pebpersistentsymbol}\bigl(#1\bigr)}
\newcommand{\bwpebpricevisiting}[1]%
    {\formatpebblingprice{BW-Peb}^{\pebvisitingsymbol}(#1)}
\newcommand{\Bwpebpricevisiting}[1]%
    {\formatpebblingprice{BW-Peb}^{\pebvisitingsymbol}\bigl(#1\bigr)}

\newcommand{\pebpricepersistent}[1]%
    {\formatpebblingprice{Peb}^{\pebpersistentsymbol}(#1)}
\newcommand{\Pebpricepersistent}[1]%
    {\formatpebblingprice{Peb}^{\pebpersistentsymbol}\bigl(#1\bigr)}
\newcommand{\pebpricevisiting}[1]%
    {\formatpebblingprice{Peb}^{\pebvisitingsymbol}(#1)}
\newcommand{\Pebpricevisiting}[1]%
    {\formatpebblingprice{Peb}^{\pebvisitingsymbol}\bigl(#1\bigr)}

\newcommand{\bwpebblingpriceempty}[1]%
    {\formatpebblingprice{BW-Peb}^{\pebvisitingsymbol}(#1)}
\newcommand{\bwpebblingpriceemptycompact}[1]%
    {\formatpebblingprice{BW-Peb}^{\pebvisitingsymbol}\bigl(#1\bigr)}

\newcommand{\stoptime}{\tau}

\newcommand{\pebdeg}{\ensuremath{d}}

\newcommand{\pebaxcompact}[2]%
        [\pebdeg]{\ensuremath{\formulaformat{Ax}^{#1} \bigl(#2 \bigr)}}

\newcommand{\pqrxvar}[6]%
    {\ensuremath{\stdnot{\varx({#1})}_{#2} \lor \stdnot{\varx({#3})}_{#4} \lor %
    \sourceclausexvar[#6]{#5}}}

\newcommand{\pqr}[6]%
    {\ensuremath{\stdnot{#1}_{#2} \lor \stdnot{#3}_{#4} \lor %
    \sourceclausenodisplay[#6]{#5}}}
\newcommand{\pqrstd}{\pqr{p}{i}{q}{j}{r}{l}}
\newcommand{\pqrall}[6]%
        {\setdescrcompact
        {\pqr{#1}{#2}{#3}{#4}{#5}{#6}}{#2,#4 \in \intnfirst{\pebdeg}}}
\newcommand{\pqrallstd}%
        {\setdescrcompact{\pqrstd}{i,j \in \intnfirst{\pebdeg}}}

\newcommand{\sourceclausexvar}[2][n]%
        {\Lor_{#1 = 1}^{\pebdeg} \varx({#2})_{#1}}
\newcommand{\subsourceclausexvar}[3][n]%
        {\Lor_{#1 = {#2}}^{\pebdeg} \varx({#3})_{#1}}

\newcommand{\sourceclausexvarnodisplay}[2][n]%
        {\textstyle \Lor_{#1 = 1}^{\pebdeg} \varx({#2})_{#1}}

\newcommand{\sourceclausenodisplay}[2][n]%
        {\textstyle \Lor_{#1 = 1}^{\pebdeg} #2_{#1}}

\newcommand{\relativisation}[1]%
    {\ensuremath{\formulaformat{Rel}\bigl(#1 \bigr)}}

\newcommand{\extendedversion}[1]{\widetilde{#1}}

%% file: notationLocal.tex
\renewcommand{\stoptime}{\tau}
\newcommand{\timet}{t}

\DeclareMathOperator{\dom}{dom}

\newcommand{\function}[3]{\ensuremath{#1\colon\!\signature{#2}{#3}}}
\renewcommand{\function}[3]{\funcdescr{#1}{#2}{#3}}

\newcommand{\cardinality}[1]{\lvert#1\rvert}

\newcommand{\PADistribution}{\mathcal{D}}
\newcommand{\chosenpset}{\mathcal{S}}

\newcommand{\RPHPnot}[3]{\formulaformat{RPHP}^{{#1},{#2}}_{{#3}}}

\newcommand{\rphpknk}{\RPHPnot{k}{n}{k-1}}

\newcommand{\ERPHPnot}[3]%
    {{\vphantom{\formulaformat{RPHP}}%
      \smash{\extendedversion{\formulaformat{RPHP}}}}^{{#1},{#2}}_{#3}}

\newcommand{\erphpknk}{\ERPHPnot{k}{n}{k-1}}

\newcommand{\ERPHP}{\formulaformat{ERPHP}}

\renewcommand{\ERPHPnot}[3]{\ERPHP^{{#1},{#2}}_{#3}}

\renewcommand{\erphpknk}{\ERPHPnot{k}{n}{k-1}}

\newcommand{\PHP}{\formulaformat{PHP}}
\newcommand{\EPHP}{\formulaformat{EPHP}}

\newcommand{\PHPk}{\PHP^{k}_{k-1}}

\newcommand{\TPHP} %
    {\ensuremath{{\vphantom{\PHP}\smash{\extendedversion{\PHP}}}^{k}_{k-1}}\xspace}

\renewcommand{\TPHP}{\ensuremath{\EPHP^{k}_{k-1}}\xspace}

\newcommand{\phpvariableremap}{\delta}

\newcommand{\TrueValue}{\ensuremath{\top}\xspace}
\newcommand{\FalseValue}{\ensuremath{\bot}\xspace}
\renewcommand{\TrueValue}{\top}
\renewcommand{\FalseValue}{\bot}

\newcommand{\PCR}{PCR\xspace}
\newcommand{\SA}{SA\xspace}
\newcommand{\SAR}{SAR\xspace}
\newcommand{\Lasserre}{Lasserre\xspace}

\newcommand{\dualvar}[1]{\ensuremath{\olnot{#1}}}
\newcommand{\dvarx}{\ensuremath{\dualvar{x}}}

\newcommand{\poly}[1]{\ensuremath{\uppercase{#1}}}
\newcommand{\mono}[1]{\ensuremath{\uppercase{#1}}}
\newcommand{\polyp}{\poly{p}}
\newcommand{\polyq}{\poly{q}}
\newcommand{\polyr}{\poly{r}}
\newcommand{\monom}{\mono{m}}
\newcommand{\polyset}[1]{\ensuremath{\mathcal{\uppercase{#1}}}}
\newcommand{\monoset}[1]{\ensuremath{\mathcal{\uppercase{#1}}}}

\newcommand{\nLit}{\ensuremath{\mathcal{J}}}
\newcommand{\pLit}{\ensuremath{\mathcal{I}}}

\newcommand{\altPHP}{\ensuremath{\formulaformat{APHP}^{k}_{k-1}}\xspace}
\newcommand{\hrank}{$H$\nobreakdash-rank\xspace}
\renewcommand{\hrank}{$H$\nobreakdash-bounded rank\xspace}
\newcommand{\hbounded}{$H$\nobreakdash-bounded\xspace}

\newcommand{\hconsistent}{$H$\nobreakdash-consistent\xspace}
\newcommand{\hdistr}[1]{\Pi_{#1}}
\newcommand{\linform}[1]{\widehat{#1}}

\newcommand{\smallassmntSAR}{\mu}
\newcommand{\largeassmntSAR}{\eta}

\newcommand{\smallmapSAR}{\varphi}
\newcommand{\largemapSAR}{\psi}

\newcommand{\YAExpectationNotation}[1][]{\ensuremath{\mathbb{E}_{#1}}}
\newcommand{\YAExpectation}[2][]{\ensuremath{\mathbb{E}_{#1}\!\left[#2\right]}}

\newcommand{\RestrictionBoundDisplay}[1]%
    {\frac{{(4k\log n)}^{k}}{n^{#1}}}
\newcommand{\RestrictionBoundInline}[1]%
    {{(4k\log n)}^{k}/{n^{#1}}}
\newcommand{\InverseRestrictionBoundDisplay}[1]%
    {\frac{n^{#1}}{{(4k\log n)}^{k}}}
\newcommand{\InverseRestrictionBoundInline}[1]%
    {n^{#1}/{(4k\log n)}^{k}}

\newcommand{\MultPoly}[1]{M(#1)}
\newcommand{\MultPolyBig}[1]{M\bigl(#1\bigr)}
\newcommand{\SumPoly}[1]{S(#1)}
\newcommand{\SumPolyBig}[1]{S\bigl( #1 \bigr)}

%% file: albertmacros.tex
\newcommand{\LS}{\mathrm{LS}}

\newcommand{\LSPLUS}{\LS^+}

\newcommand{\reals}{\mathbb{R}}

%% file: abstract.tex
\begin{abstract}
  We prove that there are $3$-CNF formulas over $n$~variables that can
  be refuted in resolution in width~$w$ but require  resolution proofs
  of size $n^{\bigomega{w}}$.  This shows that the simple counting
  argument that any formula refutable in width~$w$ must have a proof
  in size $n^{\bigoh{w}}$ is essentially tight.  Moreover, our lower
  bound generalizes
  to polynomial calculus resolution (PCR) and
  Sherali-Adams, implying that the corresponding size upper bounds in
  terms of degree and rank are tight as well.  
  Our results do not extend all the way to Lasserre, however, where
  the formulas we study have proofs of constant rank and size
  polynomial in both $n$ and~$w$.
%
%
\end{abstract}

%% file: introduction.tex
\section{Introduction}
\label{sec:introduction}

%
%

Proof complexity studies how hard it is to prove that propositional
logic formulas are tautologies. While the original motivation for this line
of research, as discussed in~%
\cite{CR79Relative},
was to prove superpolynomial lower bounds on proof size for
increasingly stronger proof systems as a way towards establishing 
$\NPclass \neq \coNP$
(and hence $\Pclass \neq \NPclass$),
it is probably fair to say that most current research in proof
complexity is driven by other concerns.

One such concern is the connection to SAT solving.  By a standard
transformation any propositional logic formula can be converted to
another formula in conjunctive normal form (CNF) that has the same
size up to constant factors and is unsatisfiable if and only if the
original formula is a tautology.  Any algorithm for solving SAT
defines a proof system in the sense that the execution trace of the
algorithm constitutes a polynomial-time verifiable witness of
unsatisfiability.%
\footnote{Such a witness is often referred to as a
  \introduceterm{refutation} rather than a \introduceterm{proof}, and
  these two terms are sometimes used interchangeably.} 
In fact, most modern-day SAT solvers can be seen to search for proofs
in systems at fairly low levels in the proof complexity hierarchy, and
upper and lower bounds for these proof systems hence give information
about the potential and limitations of the corresponding SAT solvers.
In this work, we focus on such proof systems.

\subsection{Background}

The dominant strategy in applied SAT solving today is so-called
\introduceterm{conflict-driven clause learning (CDCL)}~%
\cite{BS97UsingCSP,SS99Grasp,MMZZM01Engineering}, 
which is ultimately based on the 
\introduceterm{resolution} proof system~\cite{B37Canonical}.
The most studied complexity measure for resolution is
\introduceterm{size}
(also referred to as
\introduceterm{length}),
which gives lower bounds on the running
time on CDCL solvers and for which (optimal) exponential lower bounds
are known~%
\cite{H85Intractability,U87HardExamples,CS88ManyHard}.
Another more recently studied measure is \introduceterm{space}, 
which corresponds to memory usage, and for which (again optimal)
linear lower bounds have been proven~%
\cite{ABRW02SpaceComplexity,BG03SpaceComplexity,ET01SpaceBounds}.
For all of these results, the concept of \introduceterm{width},
measured as the size of a largest clause in a resolution proof, has
turned out to play a key role. Width was identified as a crucial
resource already in~\cite{Galil77Resolution}, and strong lower bounds on 
proof width
have been shown to imply lower bounds on
proof size~\cite{BW01ShortProofs}
and space~\cite{AD08CombinatoricalCharacterization}.

Interestingly, although the relationships and trade-offs between width
and space 
in resolution
are by now fairly well-understood~%
\cite{Ben-Sasson09SizeSpaceTradeoffs,BN08ShortProofs},
as are those between 
size 
and space~%
\cite{BN08ShortProofs,BN11UnderstandingSpace,BBI12TimeSpace,BNT12SomeTradeoffs},
very basic questions about the connections between 
size 
and width
have remained open. 
For instance, the argument in~\cite{BW01ShortProofs}
that width gives a lower bound on size works by transforming a short
resolution proof into a narrow one, but this transformation causes an
exponential increase
in the size. 
It is not
known whether such a blow-up is necessary, \ie if there are
trade-offs between 
size   
and width, or whether the analysis in~%
\cite{BW01ShortProofs}
can be sharpened to show that short
proofs can be made simultaneously narrow.
Also, as noted in the same paper, an upper bound $w$ on 
the refutation width for a
formula over~$n$ variables implies a proof 
size  
of at most
$n^{\bigoh{w}}$ simply by counting the number of possible distinct
%
clauses of width~$w$.
Again, it is not clear how tight this argument is---for all 
standard formula families
in the literature known to be refutable in 
small enough width~$w$
there are refutations in size~$n^{\bigoh{1}}$ 
independent of the width complexity
(in fact, even in size
\emph{linear} in the formula size).
To the best of our knowledge, it has been open
whether there exist formulas refutable in 
width $w = \bigoh{\sqrt{n}}$ 
that require size~$n^{\bigomega{w}}$,
\ie with the width complexity appearing in the exponent.

From a theoretical point of view, the ubiquity of CDCL in SAT solving
is somewhat puzzling since resolution is a quite weak proof system. 
A different approach is to translate CNF formulas to multilinear
polynomials and do Gröbner basis computations, which corresponds to 
\introduceterm{polynomial calculus resolution (PCR)}
as defined 
\ifthenelse{\boolean{conferenceversion}}
{in~\cite{CEI96Groebner,ABRW02SpaceComplexity}.}
{in~\cite{CEI96Groebner,ABRW02SpaceComplexity}.%
\footnote{The resolution 'R' in PCR stands for the fact that
  negated literals get their own formal variables when translating CNF
  formulas to polynomials. Such variables were missing in the original
  definition in~\cite{CEI96Groebner} but adding them makes for a more
  natural and well-behaved proof system.}}
Intriguingly, although PCR is known to be exponentially stronger than
resolution, implementations of search methods for this proof system
such as PolyBoRi~%
\cite{BD09Polybori,BDGWW09NewDevelopments}
have a hard time competing with CDCL solvers.

Proof size and space in PCR is defined in analogy with resolution, and
the measure corresponding to width of clauses is (total)
\introduceterm{degree} of polynomials. It is straightforward to show that
PCR can simulate resolution efficiently with respect to all of these
measures, meaning that
the same worst case upper bounds as in resolution apply to PCR\@.
It was proven in
\cite{IPS99LowerBounds}
that strong degree lower bounds imply strong size lower bounds,
which is a close parallel to the 
size-width 
relation for resolution in~%
\cite{BW01ShortProofs},
and this size-degree relation has been employed to prove exponential lower
bounds on  size in a number of papers, with
\cite{AR03LowerBounds}
perhaps providing the most general setting.
%
%
Optimal (linear) lower bounds on space were obtained in
\cite{BG13Pseudopartitions}
building on
\cite{ABRW02SpaceComplexity,FLNTZ12SpaceCplx},
but it is worth noting that
these bounds are \emph{not} derived from degree lower bounds---it
remains unknown whether an analogue of~%
\cite{AD08CombinatoricalCharacterization}
holds for PCR (although
\cite{FLMNV13TowardsUnderstandingPC}
recently reported some progress on this and related open questions).
Strong trade-offs between size and space as well as between degree and
space have been shown in~%
\cite{BNT12SomeTradeoffs},
but---again in analogy with resolution---the exact relations between
size and degree remains unclear. The same blow-up as in
\cite{BW01ShortProofs}
occurs in
\cite{IPS99LowerBounds}
when small size is converted to small degree, but it is not known
whether this is necessary or just an artifact of the proof.
Also, 
it was shown in~\cite{CEI96Groebner} 
that a degree upper bound of~$d$ implies proof size at
most~$n^{\bigoh{d}}$, 
but it has been open whether this is tight or not.

Yet another way to achieve greater expressivity than in resolution is
to translate clauses into linear inequalities and manipulate them
using 0-1 linear programming. Perhaps the simplest and most well-known
example of this approach is the
\introduceterm{cutting planes}
proof system introduced in
\cite{CCT87ComplexityCP}
based on ideas in
\cite{Chvatal73EdmondPolytopes,Gomory63AlgorithmIntegerSolutions}.
In this paper, however, we will be interested in somewhat related but
different \introduceterm{semialgebraic}
methods  operating on
linear programming relaxations of 
the CNF translations, such as the
\introduceterm{Sherali-Adams},
\introduceterm{Lov{\'a}sz-Schrijver},
and
\introduceterm{Lasserre}
hierarchies used for attacking \mbox{\NP-hard} optimization problems.
We discuss this next.

%
%

The  
\introduceterm{Sherali-Adams (SA)} method~%
\cite{SheraliAdams1990Hierarchy}
provides a hierarchy of linear programming relaxations of any given
$0$-$1$ integer program. The $n$th level of the hierarchy, where
$n$ is the number of $0$-$1$ integer variables, wipes out the integrality
gap and is thus exact, but also leads to an exponential blow-up in
problem size.
The main point of the method, however,  is that any linear
function of the variables can be optimized over the $k$th level of the
hierarchy in time~$n^{\bigoh{k}}$,
%
%
and in particular feasibility of the $k$th level relaxation can be
checked in that time.
In the context of proof complexity, what this means is that if the
$k$th level relaxation of the 
integer programming formulation of a CNF formula in infeasible
(the minimal such $k$ is known as the 
\introduceterm{SA~rank} of the integer program),
then there is an $n^{\bigoh{k}}$-time algorithm that can
detect this.
%
%
Furthermore, since the $k$th level of the hierarchy is an explicitly
defined linear program, its infeasibility can be certified as a
positive linear combination of its defining inequalities.  Such a
certificate is a rank-$k$ Sherali-Adams refutation of 
the corresponding CNF formula.

The 
\introduceterm{Lov{\'a}sz-Schrijver} approach~%
\cite{LovaszSchrijver1991Cones}
can be thought of as (and indeed it is formally equivalent to)
an iterated version of the level-$2$ SA~relaxation.
%
%
The point is again that any linear function can be optimized over the
linear program after $k$~iterations in time $n^{\bigoh{k}}$.
%
%
%
Lov{\'a}sz and Schrijver also introduced a method $\LSPLUS$, which uses
semidefinite programming instead of linear programming,
and which is significantly stronger in some notable cases of interest
in combinatorial optimization.

The
\introduceterm{Lasserre} method~%
\cite{Lasserre2001Explicit},
finally, 
is basically the Sherali-Adams method with semidefinite programming
conditions at all levels of the hierarchy. Again it stratifies into
levels and the $k$th level can be solved in time $n^{\bigoh{k}}$.
Moreover, Lasserre's method is the strongest of all three in the sense
that, level by level, it provides the tightest of all three
approximations of the integer linear program.  We refer 
to~\cite{Laurent2001Comparison,ChlamtacTulsiani2012Convex} for a more
detailed discussion of Sherali-Adams, Lov{\'a}sz-Schrijver and
Lasserre and 
a comparison of their relative strength.

In view of the important algorithmic applications that these methods
have (see, e.g., \cite{Parrilo00Thesis} and subsequent work), it is a
natural question whether the upper bounds $n^{\bigoh{k}}$ for
rank~$k$ are tight, just as for resolution and polynomial calculus
resolution.

%

%
From the proof complexity side, some notable early papers
investigating semialgebraic proof systems were published 
around the turn of the millennium
\cite{Pud99Complexity,GV01Complexity,GHP02Complexity},
%
%
but then this area of research 
seems 
to have gone dormant. In the last
few years,
these proof systems
have made an exciting reemergence in the context of
hardness of approximation, revealing unexpected and intriguing
connections between approximation and proof complexity.
Some examples of this is the paper
\cite{Sch08LinearLevel}
essentially rediscovering results from
\cite{Gri01LinearLowerBound},
and more recent papers such as
\cite{BBHKSZ12Hypercontractivity,OZ13Approximability}.
There have also been papers such as
\cite{BPS07LS}
and (the very recent)
\cite{GP13CommunicationLowerBounds}
focusing on \emph{semantic} versions of these proof systems, with less
attention to the actual syntactic derivation rules used.
%
%

%% file: ourresults.tex
\subsection{Our results}

The main contribution of this paper is showing that
the upper bounds on proof size in terms of width for resolution, 
degree for PCR, and rank for Sherali-Adams are essentially tight
(up to constant factors in the exponent).
Moreover, an interesting feature of our result is that we can actually
use the 
same formula family
to prove  tightness simultaneously
for all the proof systems. 
What this means is that we obtain upper bounds on size in resolution
that tightly match lower bounds in the much stronger systems PCR and
Sherali-Adams (which are in turn tight for these systems since
resolution width is an upper bound on both
PCR degree and Sherali-Adams rank).
The formal statement of this result is as follows.

\begin{theorem}
  \label{th:main-thm-informal}
  Let $w = w(n)$ be such that $w = O(n^{c})$ for 
  some positive constant
  $c < 1/2$.  
  Then there are 3-CNF formulas 
  $F_{n,w}$ 
  with 
  $\bigoh{w n}$ clauses over
  $\bigoh{n}$
  variables such that
  the following holds:
  \begin{enumerate} \itemsep=0pt
  \item 
    $F_{n,w}$ has a 
    resolution refutation
    in simultaneous size~$n^{\bigoh{w}}$, width~$\bigoh{w}$ and space~$\bigoh{w}$.
  \item 
    Any refutation of
    $F_{n,w}$ 
    in resolution, PCR, or Sherali-Adams
    must have size $n^{\bigomega{w}}$.
  \end{enumerate}
\end{theorem}

For resolution this actually shows something slightly stronger than
that the counting upper bound on size in terms of width is tight.
Namely, since the formulas in 
\refth{th:main-thm-informal}
have the same asymptotic upper bound on space as on width, it follows that
even for formulas of space complexity $\bigoh{w}$---which is a more
stringent requirement than width complexity~$\bigoh{w}$---it is still
impossible to obtain any size upper bound better than~$n^{\bigoh{w}}$ 
in general. 

\Refth{th:main-thm-informal}
has an interesting consequence for the analysis of CDCL
solver performance, which we state as a formal corollary.
By way of background, it was shown in~\cite{AFT11ClauseLearning}
that if a CNF formula           
$F$ over $n$~variables has a resolution refutation in width~$w$, 
then with high probability any CDCL solver%
\footnote{This result holds for a fairly general mathematical model of
  what a CDCL solver is, which agrees reasonably well with how
  state-of-the-art solvers are actually implemented in practice.}
will only need time~$n^{\bigoh{w}}$ to decide that $F$ is indeed
unsatisfiable.%
\footnote{Perhaps this might not seem so impressive at first sight---after all,
  exhaustive search in bounded width runs within this time bound
  deterministically---but the point is that a CDCL solver is very far
  from doing exhaustive width search and does not care at all about the
  existence or non-existence of narrow refutations.} 
An obvious question is whether this result is tight.
\Refth{th:main-thm-informal}
shows that the answer is ``yes,'' since no CDCL solver can run faster
than the shortest resolution proof it can possibly find.%
\footnote{%
  This is of course assuming that the solver does not implement
  features such as, e.g., cardinality reasoning or extended
  resolution, since these fall outside of the standard CDCL framework
  and go beyond resolution-based reasoning.}    

\begin{corollary}
  There are formulas $F$ over $n$~variables
  refutable in resolution in width~$w$
  for which any 
  resolution-based
  CDCL solver cannot run faster than
  $n^{\bigomega{w}}$, 
  and hence the result in~%
  \cite{AFT11ClauseLearning}
  is optimal up to constants in the exponent.  
\end{corollary}

Another interesting aspect of our lower bound for resolution is in
the context of Berkholz's
\EXPTIME-completeness result for deciding resolution width~%
\cite{Berkholz12ComplexityNarrowProofs}.
What Berkholz showed is that given a formula $F$ over $n$~variables and
a parameter~$w$, it cannot be decided in time less than
$n^{(w-3)/12}$
whether 
$F$ has a resolution refutation in width~$w$ or not.
Optimizing the constants in 
\refth{th:main-thm-informal},
we can show that there are
$4$-CNF formulas refutable in width~$w$
for which no resolution refutation can be shorter than
$n^{w/2 - \littleoh{1}}$. It is worth noting that this bound is
stronger than that in
\cite{Berkholz12ComplexityNarrowProofs},
although it of course applies only for the more restricted setting
where the algorithm has to output a \mbox{width-$w$} resolution refutation
rather than for the general decision problem. Still, we believe this
sheds interesting light on Berkholz's result.

\subsection{Discussion of proof techniques}
\label{sec:intro-techniques}

We conclude the overview by outlining the proof of  the lower bound in
\refth{th:main-thm-informal}
for resolution and how it differs from previously used methods.
At a high level, our proof is a standard restriction argument, but it turns
out to have some twists which we believe might be of interest and
could be useful elsewhere.%
\footnote{In fact, in a sense this has already happened in that our
  paper heavily draws on ideas from~\cite{AMO13LowerBounds}, which
  used a similar approach in a very different context.}

Before going into the details of our new restriction argument, let us
revisit previous lower bounds on size in terms of width and see how
they fall short of proving what we are after. 
On the one hand,  the result in \cite{BW01ShortProofs} states that
if a 3-CNF formula on $n$~variables
requires width $w$ to refute in resolution, then it also requires size
$2^{\bigomega{w^2/n}}$. 
This
lower bound is vacuous for
$w$ smaller than $\sqrt{n}$ and, in any case, 
can never be larger than
$2^{\bigomega{w}}$ 
since $w$ is bounded by~$n$. 
%
%
On the other hand, 
for formulas refutable in width~$w$ smaller than~$\sqrt{n}$, 
a direct random restriction argument can sometimes still
be applied to get 
meaningful lower bounds. 
The idea is that setting a random literal to true will kill off a
$\frac{w}{2n}$-fraction of the wide clauses on average. After
$r$~rounds of such restrictions, the expected number of surviving wide
clauses is at most
$\bigl(1 - \frac{w}{2n} \bigr)^r S$,
where $S$ is the size of the refutation, 
and choosing
$r = (2n/w) \log S$ brings the number of wide clauses down to zero.
A contradiction is then derived by showing that
the residual formula still requires width~$w$ to refute. Note,
however, that we cannot apply the restriction for  more than
$n$~rounds (or else there will be no residual formula to argue about), 
and so the best size lower bound this method can achieve
is again 
$2^{\bigomega{w}}$,
which is smaller than the $n^{\Omega(w)}$ bound that we are after.

In some sense, the problem is that using restrictions in the style of 
H\aa{}stad's  switching lemma~\cite{Hastad87Thesis} 
does not work in our setting. Instead, it turns out that a
seemingly weaker argument inspired by 
Furst-Saxe-Sipser~\cite{Furst1984ParityCircuits} is just what we need.
%
%
Let us now describe this modified restriction argument and how it
overcomes the problems discussed above.

We start with a carefully chosen family of formulas~$F_{n,w}$ 
and an associated distribution over random restrictions~$\rstd_n$.
Then we assume that we have a resolution refutation 
$\pi$
of~$F_{n,w}$ 
in size~$n^{\littleoh{w}}$
and analyze how a randomly chosen restriction $\rstd_n$ affects~$\pi$. 
We get two cases:
\begin{enumerate} \itemsep=0pt
\item 
  For clauses $C$ in the refutation $\pi$ that are noticeably wide, 
  $\rstd_n$ is very likely to satisfy a literal in $C$
  and so the clause disappears.
\item 
  Clauses that are not so wide will not be satisfied by~$\rstd_n$,
  but since they are reasonably small they are very likely to be
  shortened by~$\rstd$ to width strictly less than~$w$.
\end{enumerate}
Admittedly, the first case looks no different from the standard
restriction argument, and the second case 
seems 
quite weak.  But the
point is that by considering also the second case, we can afford a
significantly
bigger bound for ``wide'' 
than before,
thus getting a bigger probability of success. This is 
the key to 
our
argument. The rest is now standard: $F_{n,w}$ and $\rstd_n$ 
are
chosen so
that $F_{n,w}$ restricted by $\rstd_n$ is a 
bounded-width version of a 
pigeonhole principle (PHP)
formula with $w$ pigeons that are supposed to fit into $w-1$ holes.
Since $\pi$ is short enough, by a counting argument there is some
restriction $\rstd_n$ that eliminates all wide clauses to give a
resolution refutation of the 
PHP~formula
in width less than~$w$.  It is a straightforward separate argument
that such a narrow 
refutation cannot exist, and the lower bound on
size follows.

The lower bounds for PCR and Sherali-Adams are quite similar. The
restriction part of the argument is basically the same, but one has to
work a bit harder to prove the final punchline that the restricted
refutations have impossibly low degree and rank, respectively.

It should perhaps be stressed that while the final argument is quite
straightforward and natural
(at least for resolution), a crucial component in the proof is to
find the right formulas~$F_{n,w}$ 
and associated restrictions~$\rstd_n$
to plug into the argument, and to make a case analysis of the action
of $\rstd_n$ as above.
Both of these aspects use the techniques developed in
\cite{AMO13LowerBounds}
in an essential way.

%% file: paperoutline.tex
\subsection{Outline of this paper}

The rest of this paper is organized as follows.
After having given the necessary preliminaries in
\refsec{sec:preliminaries},
we state the main theorem for resolution and give a full proof in
\refsec{sec:proofs-const-width}.
\ifthenelse{\boolean{conferenceversion}}           
{In \refsec{sec:algebraic-semialgebraic},
  we discuss how the theorem can be strengthened to 
  polynomial calculus resolution and Sherali-Adams and why our
  approach does not work for Lasserre. We omit most of the details due
  to space constraints, however, and refer the reader to the upcoming
  full-length version for full proofs.}
{We believe this can serve as a 
  useful warm-up to the more complicated
  proofs for stronger proof systems
  that follow in 
  \refsec{sec:algebraic-semialgebraic}.
  In 
  \refsec{sec:lasserre}
  we show that
  our lower bounds do \emph{not} extend all the way to Lasserre.}
We conclude in
\refsec{sec:concluding-remarks}
with some final remarks and a discussion of open problems.

%% file: preliminaries.tex
\section{Preliminaries}
\label{sec:preliminaries}

A \introduceterm{literal} over a Boolean variable $\varx$ is either
the variable $\varx$ itself (a \introduceterm{positive literal}) or
its negation~
$\olnot{\varx}$ (a
\introduceterm{negative literal}). A \introduceterm{clause} $\clc =
\lita_1 \lor \formuladots \lor \lita_{\clwidth}$ is a disjunction of
literals. 
A \introduceterm{$\clwidth$\nobreakdash-clause} is a clause that
contains at most $\clwidth$~literals. A \introduceterm{CNF formula}
$\fstd = \clc_1 \land \formuladots \land \clc_m$ is a conjunction of
clauses.  \mbox{A \introduceterm{\kcnfform{}}} is a CNF formula consisting of
\xclause{\clwidth}{}s. We think of clauses and CNF formulas as sets:
the order of elements is irrelevant and there are no repetitions.
We denote the logical true value as $\TrueValue$ and the logical false
value as $\FalseValue$. 
The empty clause (containing no literals) is also denoted~$\emptycl$, 
since it is always false.
For integers $m$ and $n$, $m<n$, we use the standard notation 
$[n] = \{1,2,\ldots,n\}$
and $[m,n] = \{m, m+1,\ldots,n\}$.


A \introduceterm{resolution derivation} of a clause $ \clc $ from a CNF formula
$ \fstd $ is a sequence of clauses $(\clc_{1}, \ldots, \clc_\stoptime)$
such that $\clc_{\stoptime}=\clc$ and for $1
\leq \timet \leq \stoptime$ the clause $\clc_{\timet}$ is obtained by
one of the following derivation rules:
\begin{itemize} \itemsep=0pt
\item \textbf{Axiom:} $\clc_\timet$ is a clause in $\fstd$ 
  (an \introduceterm{axiom clause});
\item \textbf{Inference:} $\clc_\timet = A \lor B$, where $ \clc_{i} = A \lor
x$ and $ \clc_{j} = B \lor \olnot{x}$ for $ 1 \leq i,j < \timet $;
\item \textbf{Weakening:} $ \clc_{\timet} \supseteq \clc_{i}$ for some
$ 1 \leq i <\timet $.
\end{itemize}
%
%
A \introduceterm{resolution refutation}
of $\fstd$ is a derivation of the empty clause $\emptycl$ from $ \fstd$. 

Every resolution derivation 
$\proofstd = (C_1,\ldots,C_\stoptime)$ 
can be
associated with a directed acyclic graph 
(DAG)~$G_\proofstd$ with vertices labelled by clauses
$C_t$ in~$\proofstd$ and edges $(C_i, C_j)$
if $C_j$ is obtained by an inference or a weakening step
and $C_i$ is used as a premise in that step. 
The derivation $\proofstd$ is said to be
\introduceterm{tree-like} if 
$G_\proofstd$
is a tree.  
%
The \introduceterm{(clause) space} of $\proofstd$ at time~$t$ 
is the number of clauses derived before or at  time~$t$ that will be
used after or at  time~$t$, \ie
all clauses $C_{i}$, $i \leq t$, in~$G_\proofstd$ having an outgoing edge to
clauses $C_{j}$, $ j \geq t $
(plus the clause~$C_t$ itself).
The space of $\proofstd$ is the maximum space at any time $t$ in the
derivation. The \introduceterm{width} of $\proofstd$ is the maximum number
of literals in any clause $C_t$ in~$\proofstd$, and the 
\introduceterm{size} (or
\introduceterm{length})
of
$\proofstd = (C_1,\ldots,C_\stoptime)$ 
is~$\stoptime$. 
\ifthenelse{\boolean{conferenceversion}}{}
{We remark that 
  it is straightforward to show that all 
  applications of the weakening rule
  can be eliminated from a resolution refutation without
  any increase in size, width, or space, and while
  maintaining tree-likeness.}



In
\introduceterm{polynomial calculus resolution (\PCR)} one instead
refutes an unsatisfiable formula~$F$ over variables
$\varx_{1}, \ldots,\varx_{n}$ 
by reasoning in terms of 
polynomials in the ring
$\F[\varx_{1},\ldots,\varx_{n},\dvarx_{1},\ldots,\dvarx_{n}]$,
where $\F$ is some fixed 
field and
$\varx_{i},\, \dvarx_{i}$ are formally independent variables.
It is natural to think of polynomials as being satisfied by an
assignment when they evaluate to~$0$,  so in \PCR the truth values
$\TrueValue$ and~$\FalseValue$ are represented by $0$ and~$1$,
respectively, and a clause 
$\Lor_{i\in \pLit} \varx_{i} \lor
\Lor_{j\in \nLit}\olnot{\varx}_{j}$ 
is translated into
the 
one-term polynomial
$\prod_{i\in \pLit} \varx_{i} \, \cdot\,  \prod_{j\in \nLit}\dvarx_{j}$.
A \introduceterm{\PCR derivation} of a polynomial~$\polyr$ from a set of
polynomials $\polyset{S}=\{\polyq_1,\ldots,\polyq_m\}$ 
is a sequence 
$(\polyp_{1}, \ldots,
\polyp_\stoptime)$ 
such that $\polyp_{\stoptime}=\polyr$ and for $1 \leq \timet \leq
\stoptime$ the polynomial 
$\polyp_{\timet}$ is obtained by one of the
following 
derivation
rules:
\begin{itemize} \itemsep=0pt
\item \textbf{Boolean axiom:} 
  $\polyp_{\timet}$ is $\varx^{2}-\varx$ 
  for some variable $\varx$ (or~$\dvarx$);
\item \textbf{Complementarity axiom:} $\polyp_{\timet}$ is
  $1-\varx-\dvarx$
  for some variable $\varx$;
\item \textbf{Initial axiom:} $\polyp_{\timet}$ is one of the
  polynomials $ \polyq_{j} \in \polyset{S}$; 
\item 
  \textbf{Linear combination:} 
  $\polyp_\timet = \alpha \polyp_{i} + \beta \polyp_{j} $ for
  $ 1 \leq i,j < \timet $ and some $\alpha, \beta \in \mathbb{F}$;
\item \textbf{Multiplication:} $ \polyp_{\timet} = \varx \polyp_{i}$ 
  for $ 1 \leq i <\timet $ and some  variable $\varx$.
\end{itemize}
A \introduceterm{\PCR refutation} of $F$ is a \PCR derivation of~$1$
from the set of polynomials representing the clauses  of~$F$ as
explained above. Note that the Boolean axioms make sure that
variables can only take values
$\TrueValue=0$ and $\FalseValue=1$,
and the complementarity axioms enforce that
$\varx$
and
$\dvarx$
take opposite values.

The \introduceterm{degree} of a \PCR derivation~$\proofstd$
is the maximum of the (total) degrees of the
polynomials in~$\proofstd$. 
The \introduceterm{size} of~$\proofstd$
is the sum of the sizes of the polynomials in~$\proofstd$, 
where the size of a polynomial is defined as its number of terms.%
\footnote{Just to make terminology precise, in this paper a
  \introduceterm{monomial}
  is a product of variables, a
  \introduceterm{term}
  is a monomial multiplied by a non-zero coefficient from the field~$\F$,
  and a \introduceterm{polynomial} is a sum of terms with distinct 
  monomials.}
The \introduceterm{space} measure can also be generalized from
resolution, counting terms instead of clauses, but we will not really
need it in this paper.

    

Let us next discuss \introduceterm{semialgebraic} proof systems. 
All such proof systems encode a CNF formula as a set of polynomial 
inequalities over the reals.
A clause 
$\bigvee_{i \in \pLit} x_i \vee \bigvee_{j \in \nLit} \olnot{x}_j$
is represented by the inequality 
\mbox{$\sum_{i\in \pLit} x_{i} + \sum_{j\in \nLit} (1-x_{j}) - 1 \geq 0$,}
where we identify
\mbox{$\TrueValue=1$} and \mbox{$\FalseValue=0$}---note that this is the opposite of
the convention for \PCR.
A CNF formula 
$\fstd$
is represented by 
the inequalities corresponding to its clauses.
%
A \introduceterm{Sherali-Adams (\SA) derivation} of an inequality 
$ \polyr \geq 0 $ 
from a set of polynomial inequalities 
$\set{\polyq_{1}\geq 0, \ldots, \polyq_{m}\geq 0}$ 
is a formula of the form
\begin{equation}
  \label{eqn:SAdefinition}
  \ifthenelse{\boolean{conferenceversion}}{\textstyle}{}
  \sum^{\stoptime}_{\timet=1} \alpha_{\timet}  \cdot \prod_{i\in \pLit_{\timet}} 
  \varx_{i}\cdot \prod_{i\in \nLit_{\timet}} (1-\varx_{i}) \cdot
  \polyp_{\timet}
  \eqcomma
\end{equation}
that when expanded into a sum of terms gives
the polynomial $\polyr$,
where
$\alpha_{\timet}\in\Rplus$
and
$ \polyp_{\timet} $ is one of the 
original polynomials~%
$ \polyq_{j} $, or an
\introduceterm{axiom} of the form $ \varx^{2}_{i}-\varx_{i} $ or 
$\varx_{i}-\varx^{2}_{i} $, or the constant~$1$.
%
%
A \introduceterm{\Lasserre derivation} of 
$ \polyr \geq 0 $ 
is a formula of the form~\eqref{eqn:SAdefinition} that expands
to $\polyr$ where 
in addition
$ \polyp_{\timet} $ can be a
square
$\polyq^2$ for 
any arbitrary polynomial~$\polyq$.
%
%
Note that Sherali-Adams and \Lasserre are
\introduceterm{static} proof systems in that they have
``one-shot'' derivations, in contrast to resolution and \PCR that
construct derivations dynamically step by step.

We can augment Sherali-Adams by 
\introduceterm{twin variables} $\dvarx_{i}$ whose
intended meaning is the negation of $ \varx_{i} $, \ie
\ifthenelse{\boolean{conferenceversion}}
{$1-\varx_i$.%
\footnote{In fact, this is how \PCR was extended
  in~\cite{ABRW02SpaceComplexity} from the original definition of
  polynomial calculus (PC) in~\cite{CEI96Groebner}.}} 
{$1-\varx_i$.%
\footnote{As briefly discussed above, this is how \PCR was extended
  in~\cite{ABRW02SpaceComplexity} from the original definition of
  polynomial calculus (PC) in~\cite{CEI96Groebner}.}} 
We define a 
\introduceterm{Sherali-Adams resolution (\SAR) derivation} 
to be an \SA derivation
as in~\eqref{eqn:SAdefinition}  
except that the set of variables is
$\{\varx_1,\ldots,\varx_n,\dvarx_1,\ldots,\dvarx_n\}$ 
and  that
$ \polyp_{\timet} $ can also be a
complementarity axiom 
$1 - \varx_i - \dvarx_i$ 
or
$-1 + \varx_i + \dvarx_i$.

A 
\introduceterm{Sherali-Adams (\SA), \SAR, or \Lasserre refutation} 
of 
$ \fstd $ is a derivation in the respective system
of the
inequality $ -1 \geq 0 $ from the inequalities $ \polyq_{1} \geq 0, \ldots,
\polyq_{m}\geq 0$ that encode
the clauses of~$F$. 
The
\introduceterm{rank} of the derivation is the maximum of the degrees
among the 
polynomials to which the formulas
$\prod_{i \in \pLit_{\timet}} \varx_i \cdot \prod_{i \in
  \nLit_{\timet}} (1-\varx_i) \cdot \polyp_{\timet}$ in~\eqref{eqn:SAdefinition}
expand,
and the \introduceterm{size} of the derivation is 
the sum of the sizes of those polynomials, where again the size
of a polynomial is defined as its number of terms.

A \introduceterm{restriction}
(or \introduceterm{partial assignment})
$\rstd$ is a partial mapping from 
variables to 
$\set{\FalseValue, \TrueValue}$. 
We identify 
$\rstd$ with the set of literals it sets to true. The 
\introduceterm{domain} of \rstd\ is denoted $ \dom(\rstd) $
and   the size
of $ \rstd $ is 
$ \cardinality{\rstd} = \cardinality{\dom(\rstd)} $.
The restriction 
$\restrict{C}{\rstd}$
of a clause $C$ by~$\rstd$
is the trivial clause $\TrueValue$ if $\rstd$ sets some literal of~$C$
to true---such a clause can just be removed from any formula
or derivation---and otherwise it is the clause resulting from
deleting all literals in $C$ set to false by~$\rstd$.
The restriction~$\restrict{\fstd}{\rstd}$ 
of a CNF formula $\fstd$ is the conjunction of its restricted
clauses,
and a restricted resolution derivation
$\restrict{\proofstd}{\rstd}$
is the sequence of the restrictions of the clauses in~$\proofstd$.
It is a basic fact that 
if $ \proofstd $ is a refutation of $ \fstd $, then $
\restrict{\proofstd}{\rstd} $ is a refutation of $ \restrict{\fstd}{\rstd} $.

For
PCR derivations and the polynomials therein, restrictions are
defined similarly:
%
a restricted term vanishes if one of its variables is set to~$\TrueValue=0$ and
is otherwise  obtained by
deleting all variables set to $\FalseValue=1$,
and a restricted polynomial is the sum of its restricted terms.
Again, restrictions preserve \PCR refutations.
For \SA and \SAR, the definition is analogous except the roles of $0$
and~$1$ are reversed.

%% file: narrowproofs.tex
\section{Upper and lower bounds in resolution}
\label{sec:proofs-const-width}


\ifthenelse{\boolean{conferenceversion}}
{In this section, we state the special case of our main result for the
  resolution proof system and give a full proof. The idea is to convey the main
  ideas of the argument while avoiding the additional technical details 
  that are needed for the stronger proof systems.}  
{In this section, we establish the special case of our main result for 
  the resolution proof system.
  Although the lower bound part follows from the stronger results that
  we will prove in later sections, we believe it is instructive to
  develop the argument for resolution first.}
Let us start by 
stating a
slightly more detailed version of
\refth{th:main-thm-informal}, 
but restricted to resolution, which is what we will prove.

\begin{theorem} 
  \label{thm:largenarrow} 
  Let $k = k(n)$ be any integer-valued function
  such that $k(n) \leq n/4\log n$. 
  Then there is a family of $3$-CNF formulas $\{F_{n,k}\}_{n \geq 1}$,
  where $F_{n,k}$ has $\bigoh{n^2}$ variables and $\bigoh{kn^2}$
  clauses, such that:
  \begin{enumerate} \itemsep=0pt
  \item $F_{n,k}$ has a tree-like resolution refutation in size
    $\bigoh{k^k n^k}$, width $2k+1$, and space $2k+3$;
  \item any resolution refutation of $F_{n,k}$ has size
    $\Bigomega{\InverseRestrictionBoundInline{k-1}}$.        
  \end{enumerate}
\end{theorem}

%

Straightforward calculations show that if $k(n) =
\bigoh{n^c}$ for $c < 1$, then the upper bound is~$n^{\bigoh{k}}$ and
the lower bound is~$n^{\bigomega{k}}$.

%
%

\subsection{Definition of the formula}
\label{sec:definition-formula}


The CNF formulas we use to establish~\refth{thm:largenarrow} formalize a
\introduceterm{relativized}  
version of the 
pigeonhole principle 
which says
that there is a way to choose $k$ out
of~$n$ pigeons and send them to \mbox{$k-1$}~pigeonholes so that every
pigeon gets its own hole. More formally, the formula claims that 
there are (partial) functions $\function{p}{[k]}{[n]}$ and
\mbox{$\function{q}{[n]}{[k-1]}$} such that $p$ is one-to-one and
defined on~$[k]$, and $q$ is one-to-one and defined on the range
of~$p$.
%
Let us first describe a straightforward CNF encoding of this claim
with wide clauses that we denote $\RPHPnot{k}{n}{k-1}$.
Once the general idea is clear, we transform this into a slightly more
involved \mbox{$3$-CNF} formula which is the formula we will work
with.

The formula $\RPHPnot{k}{n}{k-1}$ 
is over
variables $p_{u,v}$ that encode the function~$p$,
$q_{v,w}$ that encode the function~$q$, 
and 
$r_{v}$
that encode a superset of the range of $p$.
It consists of the following
collection of 
\ifthenelse{\boolean{conferenceversion}}{%
  clauses, where $ u , u' $ range over~$ [k] $, $ v , v' $ range
  over $ [n] $, and $ w $ ranges over~$ [k-1] $:}
{clauses:}
\begin{subequations}
  \begin{align}
    & 
    p_{u,1} \lor p_{u,2} \lor \cdots \lor p_{u,n} &   &
    \ifthenelse{\boolean{conferenceversion}}
    {
      \text{for all $ u $,}
    }           
    {
      \text{$u\in[k]$,}
    }
    \label{eq:PdefinedWide}\\
    & \olnot{p}_{u,v} \lor \olnot{p}_{u',v} &  &  
    \ifthenelse{\boolean{conferenceversion}}
    {
      \text{for all $u \neq u'$ and $ v $,}
    }
    {
      \text{$u,u'\in[k]$, 
        $u\neq u'$, 
        $v\in [n] $,}
    }
    \label{eq:PinjectiveWide}\\
    & \olnot{p}_{u,v} \lor r_{v} &  &  
    \ifthenelse{\boolean{conferenceversion}}
    {
      \text{for all $ u $ and $ v $,}
    }
    {
      \text{$u\in[k]$,   $v\in[n]$,}
    }
    \label{eq:PimageWide}\\
    & \olnot{r}_{v} \lor q_{v,1} \lor \cdots \lor q_{v,k-1} &   &
    \ifthenelse{\boolean{conferenceversion}}
    {
      \text{for all $ v $,}
    }
    {
      \text{$ v\in[n]$,}
    }
    \label{eq:QdefinedWide}\\
    & \olnot{r}_{v} \lor \olnot{r}_{v'} \lor \olnot{q}_{v,w} 
    \lor \olnot{q}_{v',w} 
    &  &
    \ifthenelse{\boolean{conferenceversion}}
    {
      \text{for all $v \neq v'$ and $ w $.}
    }
    {
      \text{$v,v'\in[n]$, 
        $v \neq v'$,
        $ w\in [k-1]$.}
    }
    \label{eq:QinjectiveWide}
  \end{align} 
\end{subequations}
The clauses in~\eqref{eq:PdefinedWide}--\eqref{eq:PinjectiveWide} 
say that $p$ maps $[k]$ injectively into~$[n]$;
clauses~\eqref{eq:PimageWide} 
encode the range of~$p$;
and
clauses~\eqref{eq:QdefinedWide}--\eqref{eq:QinjectiveWide} 
force $q$ to be defined and injective on this range.

Next, we convert $\RPHPnot{k}{n}{k-1}$ to a $3$-CNF formula.
This is done in the standard way by using extension variables to break
up the wide clauses in \eqref{eq:PdefinedWide} and
\eqref{eq:QdefinedWide} and the \mbox{$4$-clauses}
in~\eqref{eq:QinjectiveWide}.  
%
For~\eqref{eq:PdefinedWide} we obtain the clauses
\begin{subequations}
  \begin{align}
    \label{eq:PdefinedNarrowA}
    & p_{u,1} \lor p_{u,2} \lor y_{u,2} 
    \ifthenelse{\boolean{conferenceversion}}{\eqcomma}{}
    &&  
    \ifthenelse{\boolean{conferenceversion}}
    {}
    {\text{$u \in [k]$,}} 
    \\
    \label{eq:PdefinedNarrowB}
    & \olnot{y}_{u,v} \lor p_{u,v+1} \lor y_{u,v+1}
    && 
    \ifthenelse{\boolean{conferenceversion}}
    {
      \text{for all $v \in [2,n-3]$,}
    }
    {\text{$ u\in[k]$, $v \in [2,n-3]$,}}
    \\
    \label{eq:PdefinedNarrowC}
    & \olnot{y}_{u,n-2} \lor p_{u,n-1} \lor p_{u,n}
    \ifthenelse{\boolean{conferenceversion}}{\eqcomma}{}
    && \ifthenelse{\boolean{conferenceversion}}{}{\text{$u \in [k]$,}}
  \end{align}
  \ifthenelse{\boolean{conferenceversion}}{for all $ u \in [k]$,}{}
  splitting up~\eqref{eq:QdefinedWide} 
  yields
  \begin{align}
    \label{eq:QdefinedNarrowA}
    & \olnot{r}_{v} \lor q_{v,1} \lor z_{v,1} 
    \ifthenelse{\boolean{conferenceversion}}{\eqcomma}{}
    && 
    \ifthenelse{\boolean{conferenceversion}}
    {}
    {\text{$v\in[n]$,}}
    \\
    \label{eq:QdefinedNarrowB}
    & \olnot{z}_{v,w} \lor q_{v,w+1} \lor z_{v,w+1}
    && 
    \ifthenelse{\boolean{conferenceversion}}
    {
      \text{for all $w \in [k-4]$,\;\;\;\;}
    }
    {
      \text{$ v\in[n]$, $w \in [k-4]$,}
    } 
    \\
    \label{eq:QdefinedNarrowC}
    & \olnot{z}_{v,k-3} \lor q_{v,k-2} \lor q_{v,k-1}
    \ifthenelse{\boolean{conferenceversion}}{\eqcomma}{}
    && 
    \ifthenelse{\boolean{conferenceversion}}{}{\text{$ v\in[n]$,}}
  \end{align}
  \ifthenelse{\boolean{conferenceversion}}{for all $ v \in [n]$,}{}
  and the rest of the clauses
  are
\begin{align}
  &
  \olnot{p}_{u,v} \lor \olnot{p}_{u',v}
  &  & 
  \ifthenelse{\boolean{conferenceversion}}
  {\text{for all $ u\neq u'$ and $v$,}}
  {\text{$u,u'\in[k]$, $u\neq u'$, $ v\in [n] $,}}%
  \label{eq:PinjectiveNarrow}
  \\
  & 
  \olnot{p}_{u,v} \lor r_{v}
  &  & 
  \ifthenelse{\boolean{conferenceversion}}
  {\text{for all $ u $ and $ v $,}}
  {\text{$u\in[k]$, $v\in[n]$,}}%
  \label{eq:PimageNarrow}
  \\
  &
  \olnot{r}_{v} \lor \olnot{r}_{v'}  \lor r_{v,v'} 
  &  & 
  \ifthenelse{\boolean{conferenceversion}}           
  {\text{for all $ v \neq v'$,}}
  {\text{$v,v'\in[n]$, $ v\neq v' $,}}
  \label{eq:RpairNarrow}
  \\
  &
  \olnot{r}_{v,v'} \lor \olnot{q}_{v,w} \lor \olnot{q}_{v',w} 
  &  &
  \ifthenelse{\boolean{conferenceversion}}           
  {\text{for all $v \neq v'$ and all $ w $,}}
  {\text{$v,v'\in[n]$, 
      $v \neq v'$,
      $ w\in [k-1]$.}}
  \label{eq:QinjectiveNarrow}
\end{align}
\end{subequations}
\ifthenelse{\boolean{conferenceversion}}
{where as before $ u , u' $ range over $ [k] $, $ v , v' $ range
  over $ [n] $, and $w$~ranges over~$ [k-1] $.}
{}
The \mbox{$3$-CNF} formula consisting of the clauses in
\eqref{eq:PdefinedNarrowA}--\eqref{eq:QinjectiveNarrow}, 
which we will denote~$\ERPHPnot{k}{n}{k-1}$, 
is the formula for which we will prove \refth{thm:largenarrow}.
It is easy to verify that 
this formula
has 
$\bigoh{kn^2}$~clauses over
$\bigoh{n^2}$~variables. 
We note that if we did not insist on bringing
the clause size all the way down to~$3$, then we could get a
$4$-CNF formula with
$\bigoh{kn^2}$~clauses over
$\bigoh{kn}$~variables
by not converting the $4$\nobreakdash-clauses
in~\eqref{eq:QinjectiveWide}
into the $3$\nobreakdash-clauses
\refeq{eq:RpairNarrow} 
and~%
\refeq{eq:QinjectiveNarrow}.
Our proof of
\refth{thm:largenarrow} 
works for this formula as well after straightforward adjustments and
gives a slightly better lower bound expressed in terms of the number
of variables. 

\subsection{Proof of the upper bound}
\label{sec:proof-upper-bound}

Let us 
first
describe how we can refute
the 
formula $\ERPHPnot{k}{n}{k-1}$
in resolution. 
In order to do so, we consider all sequences of the form 
$(v_{1},v_{2},\ldots,v_{k},w_{1},w_{2},\ldots,w_{k}) $, 
where 
$ v_{u}\in
[n]$ 
and 
$ w_{u}\in[k-1]$,
and the corresponding clauses
\begin{equation}
  \label{eq:mapSequenceClause}
  \Lor_{u\in[k]} \olnot{p}_{u,v_{u}} \lor 
  \Lor_{u\in[k]} \olnot{q}_{v_{u},w_{u}}.
\end{equation}
We derive all such clauses from  the axiom clauses of
$\ERPHPnot{k}{n}{k-1}$, 
and from these clauses
it is then straightforward to
obtain a
contradiction.
All of these 
derivations
are efficient, so the size of
the whole refutation is dominated by the number of clauses in~%
\eqref{eq:mapSequenceClause}.

For each clause in   \refeq{eq:mapSequenceClause} we are in one of
two cases:
either $ v_{u} =v_{u'}$ 
holds
for some \mbox{$ u \neq u'$}, 
or there must exist a pair $ v_{u} \neq v_{u'}$ with 
$ w_{u} = w_{u'} $ by the pigeonhole principle. 
In the former case, the clause~\eqref{eq:mapSequenceClause} is just a
weakening of the axiom~\eqref{eq:PinjectiveNarrow}, namely  
$ \olnot{p}_{u,v} \lor \olnot{p}_{u',v}$
with $ v=v_{u}=v_{u'} $.
In the latter case, we combine axioms 
$ \olnot{p}_{u,v_{u}} \lor r_{v_{u}}$ 
and
\mbox{$ \olnot{p}_{u',v_{u'}} \lor r_{v_{u'}}$} 
from \eqref{eq:PimageNarrow},
$ \olnot{r}_{v_{u}}\lor\olnot{r}_{v_{u'}} \lor r_{v_{u},v_{u'}}$
from \eqref{eq:RpairNarrow}, and 
$ \olnot{r}_{v_{u},v_{u'}} \lor \olnot{q}_{v_{u},w} \lor \olnot{q}_{v_{u'},w}$
from \eqref{eq:QinjectiveNarrow}, where $ w=w_{u}=w_{u'} $, 
to obtain the 
clause 
$ \olnot{p}_{u,v_{u}} \lor \olnot{p}_{u',v_{u'}} \lor \olnot{q}_{v_{u},w}
\lor \olnot{q}_{v_{u'},w}$.
It is easy to see that \eqref{eq:mapSequenceClause} can be derived
from this clause by weakening.
Since a constant number of clauses is involved in this derivation it
requires only constant space, and it is straightforward to verify
that it can in fact be carried out by a tree-like derivation in
space~$3$
(\ie keeping one clause in memory and resolving it with a
sequence of axioms). 

The rest of the refutation consists of derivations of all prefixes 
of clauses of the form \eqref{eq:mapSequenceClause} by
backward induction. For the inductive step we assume that we are able
to derive any prefix clause of  
size~$t$ in clause space $(2k-t)+3$ 
and show how to derive
any prefix of \mbox{size $t-1$} in clause space~$(2k-t+1)+3$. The
refutation ends when we reach the prefix clause of
size~$0$ (\ie the empty clause) in clause space~$2k+3$.
%
%

Suppose first that we can derive each clause of the form 
\begin{equation}
  \label{eq:prefix-clause-inductive-i}
  \Lor_{u\in[k]} \olnot{p}_{u,v_{u}} \lor 
  \Lor_{u\in[k^* - 1]} \olnot{q}_{v_{u},w_{u}}
  \lor \,
  \olnot{q}_{v^*,w^*}
  = 
  A \lor 
  \olnot{q}_{v^*,w^*}
\end{equation}
for some $k^* < k$
in clause space $s$
(writing
$v^* = v_{k^*}$
and
$w^*  = w_{k^*}$
as a shorthand). 
We want to use the existence of such derivations to derive the clause
$A$ in space~$s+1$. 
To this end,  start with the axiom
$\olnot{p}_{k^*,v^*} \lor r_{v^*}$ 
and note that the literal~$\olnot{p}_{k^*,v^*}$
also appears in the left-hand part of $A$ 
in~\refeq{eq:prefix-clause-inductive-i}.
We resolve this clause with the axiom
$\olnot{r}_{v^*} \lor q_{v^*,1} \lor z_{v^*,1}$
to get
$\olnot{p}_{u,v^*} \lor q_{v^*,1} \lor z_{v^*,1}$.
Keeping the latter clause in memory, we invoke a subderivation in
space~$s$ of the clause  
$A \lor \olnot{q}_{v^*,1}$ 
and resolve to obtain
$A \lor z_{v^*,1}$.
Continuing, assume that we have derived
$A \lor z_{v^*,w}$
for some~$w \geq 1$.
Then we resolve this clause with the axiom 
$\olnot{z}_{v^*,w} \lor q_{v^*,w+1} \lor z_{v^*,w+1}$ 
to obtain
$A \lor q_{v^*,w+1} \lor z_{v^*,w+1}$.
Keeping the latter clause in memory, 
we derive
$A \lor \olnot{q}_{v^*,w+1}$
using no more space than \mbox{$s+1$} all in all,
and then resolve to get $A \lor z_{v^*,w+1}$.
When we reach the clause
$ A \lor z_{v^*,k-3}$ we resolve it with the axiom
$\olnot{z}_{v^*,k-3} \lor q_{v^*,k-2} \lor q_{v^*,k-1}$
and then with the inductively derived clauses
$A \lor \olnot{q}_{v^*,k-2}$ 
and 
$A \lor \olnot{q}_{v^*,k-1}$
to obtain~$A$.
We point out again that the clause space of this derivation \mbox{is $s+1$.}

%

After $ k $ steps of this backward induction we get to
clauses of the form $ \olnot{p}_{1,v_{1}} \lor \olnot{p}_{2,v_{2}}
\lor \ldots \lor \olnot{p}_{k,v_{k}} $.
%
To derive the empty clause we do 
$k$~more
steps of backward
induction, 
mimicking
the procedure in the previous paragraph.
Suppose that we have shown how to derive all clauses
\begin{equation}
  \label{eq:prefix-clause-inductive-ii}
  \Lor_{u\in[k^* - 1]} \olnot{p}_{u,v_{u}}  
  \lor \, 
  \olnot{p}_{k^*,v^*}  
  = 
  A \lor  \olnot{p}_{k^*,v^*}  
\end{equation}
for $k^* < k$
and want  to derive $A$.
To do so, first resolve the
axiom $ p_{k^*,1} \lor p_{k^*,2}\lor y_{k^*,2}$ 
with the inductively derived clause 
$A \lor \olnot{p}_{k^*,1}$ 
and then with 
$A \lor \olnot{p}_{k^*,2}$ 
to get~$A \lor y_{k^*,2}$.
Suppose that we have shown how to derive
$A \lor y_{k^*,v}$
in this way for $v \geq 2$.
In order to obtain
$ A \lor y_{k^*,v+1}$  
we resolve 
$\olnot{y}_{k^*,v} \lor p_{k^*,v+1} \lor y_{k^*,v+1}$ 
with 
$A \lor y_{k^*,v}$ 
and then with~%
$A \lor \olnot{p}_{k^*,v+1}$.
We iterate up to $ A \lor y_{k^*,n-2}$
and finally resolve the axiom 
$\olnot{y}_{k^*,n-2} \lor p_{k^*,n-1} \lor p_{k^*,n}$ 
with the clauses 
$A \lor y_{k^*,n-2}$, 
$A \lor \olnot{p}_{k^*,n-1}$, 
and 
$A \lor \olnot{p}_{k^*,n}$
to obtain~$A$. 
After $k$~steps of this second stage we reach the empty clause and
the refutation is complete. 
As before, the clause space goes up by an additive one for every
inductive step, so the clause space of the whole refutation 
\mbox{is $2k+3$.}

To analyze the size of the resolution refutation obtained in this way,
consider the prefix tree of the sequences 
$(v_{1},v_{2},\ldots,v_{k},w_{1},w_{2},\ldots,w_{k})$.
Each vertex of this tree corresponds to one of the clauses $ A $ 
derived during the backward induction, with the empty clause at the
root and clauses (\ref{eq:mapSequenceClause}) at the leaves.  The
length of the derivation of each clause is linear in the number of
children, and in addition we derived the leaves with a constant number of
steps. Therefore we can charge a constant amount of steps per vertex.
The size of the tree is $O(k^{k}n^{k})$, and it follows that this is
also the size of
the refutation. 
Furthermore, the refutation is tree-like since no intermediate
clause is used more than once.   
One can also observe that 
the width of the refutation is $2k+1$ and 
reaches this maximum at the induction step from sequences of 
\mbox{length $2k$} to sequences of \mbox{length $ 2k-1 $.}

\subsection{Proof of the lower bound for resolution}
\label{sec:proof-lower-bound}

As discussed in
\refsec{sec:intro-techniques},
we use a random restriction argument to prove our size lower bound for
resolution refutations of the formula
$\ERPHPnot{k}{n}{k-1} $.
We define a distribution $\PADistribution$ on partial assignments $\rstd$ by
picking a 
subset  $\chosenpset= \set{v_{1},v_{2},\ldots,v_{k}}$ of $k$~%
elements from~$[n]$ uniformly at random
and 
letting 
$\rstd$
assign values to variables as follows:
\begin{itemize} \itemsep=0pt
\item
  $ r_{v} = \top $ for all $v \in \chosenpset$; $ r_{v}=\bot $ otherwise;
\item
  $ r_{v,v'} = r_{v} \land r_{v'}$ for all $v \neq v'$;
\item
  $ p_{u,v_{u}} = \top $
  and
  $p_{u,v} = \bot $ 
  for 
  all $ u\in[k]$
  and all $v \neq v_u$;
  %
%
\item
  $ y_{u,v} $ for all $u$ and $v$
  are set arbitrarily 
  so as   to 
  satisfy the clauses~%
  \eqref{eq:PdefinedNarrowA}--\eqref{eq:PdefinedNarrowC};
\item
  $ q_{v,w} $ and $ z_{v,w} $ are left 
  unset for all
  $ v \in \chosenpset$ and all $w$;
\item
  $ q_{v,w}=b_{v}$ and $ z_{v,w}=b_{v} $ 
  for
  all $ v \in [n]\setminus \chosenpset $ and
  all $w\in[k-1]$,
  where
  $b_{v}\in\set{\bot,\top}$ 
  is chosen uniformly and independently at random for every
  $ v \in [n]\setminus \chosenpset $.

\end{itemize}
We want to argue that with high probability such restrictions
remove or at least significantly shrink wide clauses.


For $v \in [n]$,
let us say that the variables
$\set{ q_{v,1},\ldots,q_{v,k-1},z_{v,1},\ldots,z_{v,k-1} }$
\introduceterm{mention the pigeon~$v$}. 
We say that a clause (or term)
mentions $v$ if it contains some variable in this set and define the  
\introduceterm{pigeon-width} 
to be the number of pigeons mentioned.
%
%
The next lemma describes the effect of random restrictions
$\rstd$ from
$\PADistribution$  on clauses (or terms)
depending on their pigeon-width.
Namely, a sufficiently wide clause, \ie mentioning a lot of pigeons, is
satisfied by the random restriction with high probability, 
whereas a narrower clause 
may not have its truth value fixed by
the restriction but will
with high probability contain few pigeons afterwards.

\begin{lemma}\label{lmm:StrongRestriction}
  Let $k,\ell,n$ be natural numbers 
  such that 
  $ n\geq 16$ and $\ell \leq k \leq n/4\log n$.
  Let $\cla$ be either a 
  clause or term
  over
  the variables
  of $\ERPHPnot{k}{n}{k-1}$ and let $ \rstd $ be a random restriction
  sampled from
  the distribution~$\PADistribution $
  as defined above. 
  Then the pigeon-width of 
  $\restrict{\cla}{\rstd}$ 
  is less than
  $\ell$ 
  with probability at least
  $1 - \RestrictionBoundInline{\ell}$.
%
%
\end{lemma}

\begin{proof}
%
  Let us assume that $\cla$ is a clause---the proof for terms
  (which will be used for \PCR  and Sherali-Adams)
  is completely analogous.
  Let $ v_{1},\ldots,v_{r} $ be the pigeons mentioned in $ \cla $
  sorted in some order and let 
  $ \lita_{1}, \ldots, \lita_{r}$ 
  be a sequence of literals such that $ \lita_{i}$ witnesses
  that $ \cla $ mentions~$ v_{i} $.
%
%

  If $r > 2k \log n$, then the probability that the clause~$A$ is
  not satisfied by the restriction is at most
  \ifthenelse{\boolean{conferenceversion}}{%
    \begin{multline}
      \prod^{r}_{i=1} 
      \CondProb{\rstd(\lita_{i})\neq \top}
               {\rstd(\lita_{j})\neq \top \text{ for $ j<i $}} 
      \leq 
      \\ 
      \leq 
      \prod^{r}_{i=1} 
      \CondProb{\rstd(\lita_{i})\neq \top}
               {v_j \notin \chosenpset \text{ for $ j<i $}} 
      \leq 
      \\ 
      \leq 
      \prod^{r}_{i=1}{\left( \frac{1}{2} + \frac{k}{n-i}\right)} 
      < {\left( \frac{5}{8} \right)}^{2k \log n} < \frac{1}{n^{k}}
      \eqperiod
    \end{multline}%
  }{%
    \begin{align}
      \nonumber
      \Prob{\rstd(\lita_{i})\neq \top \text{ for all } i = 1, \ldots, r}
      &\leq
      \Prob{\rstd(\lita_{i})\neq \top \text{ for all } i = 1, \ldots, \ceiling{2k \log n}}
      \\
      \nonumber
      &=
      \prod_{i=1}^{\ceiling{2k \log n}}
      \CondProb{\rstd(\lita_{i})\neq \top}
               {\rstd(\lita_{j})\neq \top \text{ for $ j<i $}} 
      \\
      &\leq 
      \prod_{i=1}^{\ceiling{2k \log n}}
      \CondProb{\rstd(\lita_{i})\neq \top}
               {v_j \notin \chosenpset \text{ for $ j<i $}} 
      \\ 
      \nonumber
      &\leq 
      \prod_{i=1}^{\ceiling{2k \log n}}
      \left( 
      \frac{1}{2} + \frac{k}{n-i}
      \right) 
      \\
      \nonumber
      &< 
      {\left( \frac{5}{8} \right)}^{2k \log n} < 
      \ \frac{1}{n^{k}}
      \eqperiod
    \end{align}}
  To see this, note that the event
  $\rstd(\lita_{i})\neq \top$
  occurs either if the pigeon $ v_{i} $ is not picked or if the 
  literal~$\lita_{i}$ is set to the wrong value. 
  %
  Assuming that no pigeon
  $v_1, \ldots, v_{i-1}$
  has been picked before~$v_i$, the conditional probability of $v_{i}$ being
  included in~$\chosenpset$ is   $k/(n-i)$, and is less otherwise.
  If $v_i \in \chosenpset$, then $\lita_{i}$ gets the wrong value with
  probability $1/2$.  
  %
  %
  The final inequalities hold because the
  ratio $ k/(n-2k \log n) $ is at most $1/(2\log n)$, and therefore it is
  at most $ 1/8 $ for $n\geq 16$.

  If instead the number of pigeons mentioned by~$A$ is
  $ r \leq 2 k \log n $, 
  we want to bound the probability that there are 
  at least $\ell$ pigeons mentioned in~$\cla$ that are chosen
  in~$\chosenpset$ 
  and hence 
  survive. 
  The choices of $\chosenpset$ with exactly $i$ pigeons mentioned
  in~$\cla$ are $\binom{r}{i}\binom{n-r}{k-i}$.
%
  Considering all possible intersections of size at
  least $\ell$ between 
  the set $\chosenpset$ and the $r$~pigeons mentioned
  in~$\cla$,
  we obtain that the probability of 
  $\ell$ surviving pigeons is at most
  \ifthenelse{\boolean{conferenceversion}}             
  {%
    \begin{multline}
      \sum^{k}_{i=\ell}
      \left.
      \binom{r}{i}\binom{n-r}{k-i}
      \right.\!\! \left/
      \binom{n}{k}
      \right. 
      \leq 
      \\ 
      \leq 
      \left.
      k \binom{\lfloor{2k \log n}\rfloor}{k} \binom{n}{k-\ell}
      \right.\!\! \left/
      \binom{n}{k} 
      \right.
      \leq 
      \frac{{(3 k^{2} \log n)}^{k}}{n^{\ell}} \eqperiod
    \end{multline}%
  }{%
    \begin{align}
      \nonumber 
      \sum^{k}_{i=\ell}
      \binom{r}{i}\binom{n-r}{k-i}{\binom{n}{k}}^{-1} 
      &\leq 
      k \binom{\lfloor{2k \log n}\rfloor}{k} \binom{n}{k-\ell}
      {\binom{n}{k}}^{-1} 
      \\
      & \leq
      \frac{k (2k \log n)^{k}}{k!} 
      \cdot
      \frac{n!}{(k-\ell)!(n-k+\ell)!}
      \cdot
      \frac{(n-k)!k!}{n!}
      \\
      \nonumber 
      & \leq
      k (2k \log n)^{k}
      \cdot
      \frac{1}{(k-\ell)!}
      \cdot
      \frac{1}{{(n-k)}^{\ell}} 
      <
      \frac{k (2k \log n)^{k}}{{(n-k)}^{\ell}} \eqperiod
    \end{align}%
    To finish
    the computation we use that $n \geq 16$ and $k \leq
    n/4\log n$ to get that $k \leq n/16$, and we observe that
    $k{(16/15)^{\ell}}\leq 2^{k}$ for every $1\leq \ell \leq k$. We
    obtain that
    \begin{equation}
      \frac{k (2k \log n)^{k}}{{(n-k)}^{\ell}}
      \leq
      \frac{k {(2k \log n)}^{k}}
           {(15n / 16 )^{\ell}} 
      =
      k(16/15)^{\ell} \cdot \frac{{(2k \log n)}^{k}}{n^{\ell}}
      \leq
      \RestrictionBoundDisplay{\ell}\eqperiod
    \end{equation}}
  This concludes the proof.
\end{proof}

We can use
\reflem{lmm:StrongRestriction} 
to show that if we hit a sufficiently short
resolution refutation of $\ERPHPnot{k}{n}{k-1}$
with a random restriction~$\rstd$, 
then in the restricted refutation all clauses 
are likely to have small pigeon-width. 
The reason this is useful is that 
the distribution~$\PADistribution$
is constructed so that 
the restricted
formula is just the standard pigeonhole principle formula, or rather,
a $3$-CNF version of it (up to renaming of variables).
To spell this out explicitly,
after renaming the $k$~pigeons in~$[n]$ 
chosen by $\rstd$ to $1,\ldots,k$, 
what remains is the following collection of
\ifthenelse{\boolean{conferenceversion}}{%
  clauses (with $ v , v' $ ranging
  over $ [k] $ and $ w $ ranging over $ [k-1] $ unless stated otherwise):}
{clauses:}
\begin{subequations}
  \begin{align}
    &  q_{v,1} \lor z_{v,1} 
    && \ifthenelse{\boolean{conferenceversion}}{\text{for all $ v$,}}{\text{$ v\in [k]$,}} 
    \label{eq:QdefinedRestrictedA} 
    \\
    & \olnot{z}_{v,w} \lor q_{v,w+1} \lor z_{v,w+1}
    && \ifthenelse{\boolean{conferenceversion}}{\text{for all $ v $ and $w \in [k\!-\!4]$,}}{\text{$ v \in [k]$, $w \in [k-4]$,}} 
    \label{eq:QdefinedRestrictedB} 
    \\
    & \olnot{z}_{v,k-3} \lor q_{v,k-2} \lor q_{v,k-1}
    && \ifthenelse{\boolean{conferenceversion}}{\text{for all $ v $,}}{\text{$ v\in [k]$,}}
    \label{eq:QdefinedRestrictedC} 
    \\
    &       
    \label{eq:QinjectiveRestricted} 
    \olnot{q}_{v,w} \lor \olnot{q}_{v',w}
    & &  \ifthenelse{\boolean{conferenceversion}}{\text{for all $
        v\neq v' $ and  $ w $.}}{\text{$v,v'\in [k]$, $v \neq v'$,  $ w\in [k-1]$.}}
  \end{align}
\end{subequations} 
But the clauses
\refeq{eq:QdefinedRestrictedA}--\refeq{eq:QinjectiveRestricted},
which we will denote $\TPHP$, 
can easily be shown to require
almost
maximal pigeon-width in resolution.
%

\begin{lemma}\label{lmm:HardRestrictedFormula}
  Every resolution refutation of $\TPHP$ has pigeon-width at least
  $k-1$.
\end{lemma}

\begin{proof}
%
  We use a game argument in the style
  of~\cite{Pudlak00ProofsAsGames,AD08CombinatoricalCharacterization}
  adapted to the notion of pigeon-width. The game is played between a
  \emph{prosecutor} and a \emph{defendant}. At each step of the game
  the prosecutor queries the defendant for the value of a variable of
  \TPHP and stores the answer in his record. The prosecutor is also
  allowed to erase variable assignments from his record after any
  query, but if so the defendant can answer differently next time
  she is asked
  about an erased variable.
  The goal of the prosecutor is to force the defendant to
  falsify a clause from $\TPHP$, while the goal of the defendant is
  to answer queries without falsifying any axiom clause in 
  this 
  formula.

  To establish the lemma, it is sufficient to show that the prosecutor
  cannot win unless at some point he holds a record that mentions $k$
  pigeons. The reason for this is that if there exists a   resolution
  refutation $\proofstd$ of pigeon-width 
  $\ell < k-1$, 
  then the prosecutor can use
  such a refutation to construct a strategy that never
  mentions more than 
  $\ell + 1$ pigeons. 

  To build a winning strategy from a refutation~$\proofstd$, 
  the prosecutor walks backwards through 
  the associated graph~$G_\proofstd$ 
  from the final empty clause all the way to some
  axiom clause. The invariant maintained is that at each step the 
  current assignment on record is the minimal falsifying assignment
  for the clause currently visited in~$G_\proofstd$.
  %
  %
  At the beginning of the game the empty
  record corresponds to the empty clause in the refutation.
  If the current clause was obtained by resolution,   the
  prosecutor queries the resolved variable 
  (which might temporarily increase 
  the number of pigeons on record
  by~$1$), 
  moves to the premise falsified by the answer, and then forgets all
  assignments not needed to falsify that clause. For a weakening step,
  the prosecutor just needs to forget variables.  The prosecutor wins
  when the game reaches a source vertex in~$G_\proofstd$ (if not
  earlier), since by the invariant the corresponding axiom clause is
  falsified by 
  the assignment on record at that point.

  Switching to the lower-bound perspective, 
  let us now briefly
  describe a defendant strategy that works 
  against prosecutors mentioning less than $k$~pigeons.
  The defendant privately keeps
  a partial matching of the pigeons mentioned in the 
  current record of the prosecutor into holes, making
  sure that this mapping is compatible with the partial assignment in
  his record.
  If the prosecutor asks about a variable   which mentions 
  a pigeon already in the domain of the defendant's partial matching,
  she
  answers consistently with 
  her  matching.
  If the prosecutor erases all variables mentioning a pigeon, the
  defendant removes that pigeon from the partial mapping, freeing up
  the  corresponding hole for later reuse.  
  If the prosecutor queries a
  variable that mentions a new pigeon, we are in one of two cases:
  either there is at least one free hole, or the record mentions $k-1$
  pigeons. In the first case the defendant assigns the new pigeon to some
  free hole and updates 
  her
  partial matching accordingly. In the
  second case the defendant has achieved 
  her goal---although she is now forced to falsify a clause of $\TPHP$
  and loses, the prosecutor
  was able to 
  win only by 
  compiling a record that 
  mentions $k$~pigeons.
%
\end{proof}

Putting all the pieces together we can now prove the lower bound in
\refth{thm:largenarrow}.   
Namely, let $\proofstd$ be a resolution refutation of $\ERPHPnot{k}{n}{k-1}$
of size $ S $. Hit $\proofstd$ with a random restriction
$\rstd$ distributed according to $\PADistribution$. Since resolution
refutations are preserved under restrictions,
$\restrict{\proofstd}{\rstd}$ is a refutation of
$\restrict{\ERPHPnot{k}{n}{k-1}}{\rstd}$ which, as discussed above, is
$\TPHP$ after renaming of variables.  By
\reflem{lmm:HardRestrictedFormula}, this refutation must have
pigeon-width 
at least~$k-1$
with probability~$1$. 
On the other hand, 
using
\reflem{lmm:StrongRestriction} with 
$\ell = k-1$ 
and taking a  union bound over
all clauses in~$\proofstd$, the probability that this happens is
at most 
$S \cdot \RestrictionBoundInline{k-1}$ 
for large enough~$n$. We
can hence
conclude that 
$S \geq \InverseRestrictionBoundInline{k-1}$, 
and the proof of~\refth{thm:largenarrow} is complete.

%% file: algebraicandsemialgebraic.tex
\section{Algebraic and semialgebraic proof systems}
\label{sec:algebraic-semialgebraic}

%
%

Let us now show how the  size lower bound for resolution in
\refsec{sec:proofs-const-width} can be generalized to polynomial
calculus resolution (\PCR) and Sherali-Adams resolution (\SAR).    
The overall structure of the size lower bound proof is very similar to
that for resolution in that we first establish a lower bound on a
parameter analogous to the pigeon-width in 
\refsec{sec:proofs-const-width}, which we call
\introduceterm{pigeon-degree} for \PCR and \introduceterm{pigeon-rank}
for \SAR, and then plug this bound into the random restriction argument as in
the proof of \reflem{lmm:StrongRestriction}.

In this section, we also discuss how upper bounds for 
\PCR and \SAR analogous to those for resolution in
\refth{th:main-thm-informal}
can be established.
%
%
The upper bound in resolution more or less immediately carries over to
\PCR, in the sense that it is very easy to show that a resolution
refutation can be simulated easily in \PCR in essentially the same
size and with \PCR degree matching the resolution width.
For \SA and \SAR it requires a bit more work to construct such
efficient simulations  and we discuss it in some detail below.
%
It should be noted that
while  \PCR degree and \SAR rank upper bounds
$\bigoh{k}$
are sufficient to obtain
refutation of size $n^{\bigoh{k}}$ in both proof systems, using explicit
simulations like the ones discussed 
in this section
gives better bounds.

\subsection{Lower bound on degree for polynomial calculus resolution}
\label{sec:lower-bound-degree}

%
%
%
%

In a natural generalization of the terminology in
\refsec{sec:proofs-const-width}, we say that not only the
variables~$q_{v,w}$ and~$z_{v,w}$ of $\TPHP$ but also their
twins~$\dualvar{q}_{v,w}$ and~$\dualvar{z}_{v,w}$ \emph{mention} the
pigeon~$v$.  The \emph{pigeon-degree} of a monomial is the number of
pigeons that are mentioned by its variables, the pigeon-degree of a
polynomial is the maximum pigeon-degree of its monomials, and
the pigeon-degree of a PCR refutation of $\TPHP$ is the maximum
pigeon-degree of the polynomials in the refutation. The following lower
bound for pigeon-degree of PCR refutations is the analogue of
\reflem{lmm:HardRestrictedFormula} for resolution.

\begin{lemma}
  \label{lmm:pcr-pigeon-degree}
  Every \PCR refutation of $\TPHP$ has pigeon-degree
  at least $\lceil{\frac{k-1}{2}}\rceil$.
\end{lemma}

\begin{proof}
%
%
We prove the lower bound by studying a different encoding $\altPHP$ of
the pigeonhole principle for $k$ pigeons and $k-1$ holes described
in~\cite{Razborov98LowerBound}. Given any \PCR refutation of $\TPHP$
as defined in
\refeq{eq:QdefinedRestrictedA}--\refeq{eq:QinjectiveRestricted}
in which all monomials mention at most $d$~pigeons, we show how to
transform it into a refutation of degree~$d+1$ of~$\altPHP$.  Since
$\altPHP$ requires degree strictly larger than $\lceil \frac{k-1}{2}
\rceil$ by Theorem~3.9 in~\cite{IPS99LowerBounds}, it follows that 
$ d \geq \lceil \frac{k-1}{2} \rceil $.

The alternative formulation $\altPHP$ is defined on variables
$x_{v,w}$ for $v \in [k]$ and $w \in [k-1]$, where $x_{v,w}=1$ means
that pigeon $v$ sits in hole $w$. We stress that this 
interpretation of the variables is the opposite of the one we use 
for~$\TPHP$. Also, $\altPHP$ is not a (translation of~a) CNF formula
but consists of the following polynomials: 
\begin{subequations}
  \begin{align}
    & 1-\sum_{w \in [k-1]}x_{v,w} 
    && \text{$ v\in [k]$,} 
    \label{eq:altPHPdefinedRestrictedA} 
    \\
    & x_{v,w}x_{v',w}
    && \text{$ w \in [k-1]$, $v,v' \in [k]$, $ v\neq v' $,} 
    \label{eq:altPHPdefinedRestrictedB} 
    \\
    & x_{v,w}x_{v,w'}
    && \text{$ v \in [k]$, $w,w' \in [k-1]$, $ w\neq w' $.} 
    \label{eq:altPHPdefinedRestrictedC} 
  \end{align}
\end{subequations} 
To 
obtain
a degree-$(d+$$1)$ refutation for $ \altPHP $, the first step
is to apply a substitution~$\phpvariableremap$ 
to the variables in the refutation
in  pigeon-degree~$ d $ of~$\TPHP$.
For $q$\nobreakdash-variables
we
define $\phpvariableremap(q_{v,w}) = 1 - x_{v,w}$ and
$\phpvariableremap(\olnot{q}_{v,w}) = x_{v,w}$, and for
$z$\nobreakdash-variables 
we let
$\phpvariableremap(z_{v,w}) = 1 - \sum_{j>w}
x_{v,j}$ and $\phpvariableremap(\olnot{z}_{v,w}) = 1 - \sum_{j \leq
  w} x_{v,j}$.
This substitution transforms the refutation of %
$\TPHP$ into a sequence of polynomials 
over the variables in~$ \altPHP $. 
%
This is not yet a valid refutation, however, and in order to deal with this
we need to show how
to derive each substituted polynomial in the sequence. 
How to do so
depends
on what rule was used to derive the polynomial before the
substitution. 

For inference steps, 
if we derived $ x \polyp $ from $ \polyp $ then 
$\phpvariableremap(x\polyp)=\phpvariableremap(x)\phpvariableremap(\polyp)$ 
can
be derived from $ \phpvariableremap(\polyp) $ 
by 
a sequence of
multiplications and linear combinations, 
and if 
the polynomial was derived
via a linear combination, then the same derivation step is valid for
the substituted polynomials.

If $\polyp$ is an application of the Boolean axiom
$x^2 - x$
to a~$q$\nobreakdash-variable
or  $z$\nobreakdash-variable,
then $\phpvariableremap(\polyp)$
can be derived from Boolean axioms combined with
polynomials~\refeq{eq:altPHPdefinedRestrictedC}.
Applications of complementarity axioms are either vacuous (for
$q$\nobreakdash-variables)
or reduce to 
\refeq{eq:altPHPdefinedRestrictedA}
(for $z$\nobreakdash-variables).

%
%
%


Finally, we need to show how to derive
$\phpvariableremap(\polyp)$
if $\polyp$ is obtained from one of the clauses
\mbox{in~\refeq{eq:QdefinedRestrictedA}--\refeq{eq:QinjectiveRestricted}}.
We describe how to do this for
$\polyp = \olnot{z}_{v,w}q_{v,w+1}z_{v,w+1}$
as in
\refeq{eq:QdefinedRestrictedB};
the other cases 
are very similar. We have
\begin{equation}
  \begin{split}
    \phpvariableremap(\olnot{z}_{v,w}q_{v,w+1}z_{v,w+1}) 
    &= 
    \left( 1-\sum_{j \leq w} x_{v,j} \right) 
    \Big(1-x_{v,w+1} \Big)
    \left( 1-\sum_{j > w+1} x_{v,j} \right)
    \\
    &= 
    1 - \sum_{j \in [k-1]} x_{v,j} 
    + \sum_{\substack{j \neq j', \, j \leq w+1, \\ j' \geq w+1 }} 
    x_{v,j} x_{v,j'} \polyr_{j,j'} 
    \eqcomma
  \end{split}
\end{equation}
where 
$1 - \sum_{j \in [k-1]} x_{v,j}$ 
is~\refeq{eq:altPHPdefinedRestrictedA}
and
all polynomials
$x_{v,j} x_{v,j'} \polyr_{j,j'}$
can be derived by multiplications and linear combinations from~%
\refeq{eq:altPHPdefinedRestrictedC}.
%
%
Thus, $\phpvariableremap(\olnot{z}_{v,w}q_{v,w+1}z_{v,w+1}) $ can be
derived from~$ \altPHP $.
   
This shows how we can apply the substitution~$\phpvariableremap$ to a
refutation of~$\TPHP$ to obtain a refutation of $ \altPHP$. The
substitution exchanges variables indexed by the pigeon $v$ for
\mbox{degree-$1$} polynomials which mention just~$v$, and therefore each
monomial of this refutation mentions at most $ d $ pigeons as well.
We then postprocess the refutation of~$\altPHP$ by removing all
monomials that  mention the same pigeon twice or more and all the monomials that mention  more than one pigeon for the same hole.
This is possible using the axioms 
$x_{v,w}x_{v',w}$
in~\refeq{eq:altPHPdefinedRestrictedB}
and
$x_{v,w}x_{v,w'}$
in~\refeq{eq:altPHPdefinedRestrictedC},
and as a result we obtain a refutation of (total) degree at
most~$d+1$. The lemma follows. 
\end{proof}

\subsection{Size and rank upper bounds for Sherali-Adams refutations}
\label{sec:size-rank-upper-bound-sa}

%
%

Let us next switch focus to upper bounds and show that
\SAR can simulate resolution refutations efficiently
in term of size and rank. We remark that  a similar simulation is given
\cite{DantchevMartinRhodes2009Tight}, but since that paper uses a
slightly different definition of Sherali-Adams we give a
full description of the simulation here for completeness.

%
%

We start by introducing notation for two polynomial forms which we
will use to represent clauses.
For any pair of sets of propositional variables 
$Y, Z$, $Y \intersection Z = \emptyset$,
we let
\begin{align}
  \SumPoly{Y,Z}  &= \sum_{y \in Y} y + \sum_{z \in Z} \olnot{z} \\
\intertext{and}
  \MultPoly{Y,Z} &= \prod_{y \in Y} \olnot{y} \prod_{z \in Z} z 
  \eqperiod
\end{align}
Consider a clause 
$C
=
\Lor_{y \in V_C^{+}} y
\lor
\Lor_{z \in V_C^{-}} \olnot{z}
$
where
$V_C^{+}$ and $V_C^{-}$ 
are the sets of variables appearing positively and negatively in~$C$, 
respectively. 
%
%
Then we define 
\begin{align}
  \SumPolyBig{\clc} &= \SumPoly{V_{\clc}^{+},V_{\clc}^{-}} \\
  \intertext{and}
  \MultPolyBig{\clc} &= \MultPoly{V_{\clc}^{+},V_{\clc}^{-}} 
  \eqperiod
\end{align}
Observe that for any assignment of the variables to 
\mbox{$\TrueValue=1$} or \mbox{$\FalseValue=0$} it holds that
$\SumPoly{\clc} - 1 \geq 0$
and
$-\MultPoly{\clc}\geq 0$ 
if and only if $\clc$ is satisfied.
The former, additive inequality is how clauses are translated to
inequalities as discussed in 
\refsec{sec:preliminaries},
but for our simulation of resolution by \SAR we will need to work with
the latter, multiplicative version.

The following three lemmas show how to efficiently simulate the steps
in a resolution derivation.


\begin{lemma}[Simulation of axiom]
  \label{lem:SARdownload}
  For a clause $\clc$ of width~$w$ the inequality
  $ - \MultPoly{\clc}\geq 0$ has 
  a derivation in \SAR of rank $w+1$ and size $O(w^{2})$ from the
  inequality 
  $\SumPoly{\clc}-1 \geq 0$.
\end{lemma}
\begin{proof}
  If $\clc$ is the empty clause then the claim is obvious since in that
  case $-\MultPoly{\clc} = \SumPoly{\clc} - 1$.
  %
  %
  Let $\clc$ be non-empty and assume for simplicity that
  it has a positive literal $x$.
  Then $\clc$ has the form
  \begin{equation}
    x \lor \Lor_{y \in Y} y \vee \Lor_{z \in Z} \olnot{z}
  \end{equation} with $ |Y|+|Z|<w $.
  By multiplying 
  $\SumPoly{\clc} - 1 \geq 0$ by $\MultPoly{Y,Z}$ 
  we
  obtain
  the polynomial inequality
  \begin{equation}\label{eq:SARaxiomstep1}
    x \MultPoly{Y,Z}
    +
    \sum_{y \in Y} y \MultPoly{Y,Z}
    + 
    \sum_{z \in Z} \olnot{z} \MultPoly{Y,Z}
    -
    M(Y,Z) \geq 0\eqperiod
  \end{equation}
  For each $y \in Y$ we
  can derive
  \begin{multline}
    \label{eq:SARaxiomstep2}
    (1-y-\olnot{y}) \cdot
    y \MultPoly{Y\setminus\set{y},Z} 
    + (y^{2} - y) \cdot 
    \MultPoly{Y\setminus\set{y},Z}
    =  
    \\
    -y \cdot
   \olnot{y}\cdot\MultPoly{Y\setminus\set{y},Z}
    = 
    -y \cdot
    \MultPoly{Y,Z}
    \eqperiod
  \end{multline} 
  In essentially the same way we can derive 
  \begin{equation}
    \label{eq:extra-SAR-inequality}
    -\olnot{z} \MultPoly{Y,Z} \geq 0    
  \end{equation}
  for each $z \in Z$.
  The inequality 
  \begin{equation}
    -\MultPoly{\clc} = -\olnot{x}\MultPoly{Y,Z} \geq 0
  \end{equation}
  is now the sum of
  inequality \eqref{eq:SARaxiomstep1}, all inequalities of the form
  \eqref{eq:SARaxiomstep2}
  and
  \eqref{eq:extra-SAR-inequality}
  for all $y \in Y$ and $z \in Z$, and of the
  inequality
  \begin{equation}
    (1-x-\olnot{x})
    \MultPoly{Y,Z} \geq 0
    \eqperiod
  \end{equation}
  This \SAR derivation has size $O(w^{2})$ and rank $w$.
\end{proof}

\begin{lemma}[Simulation of weakening]
  \label{lem:SARweakening} 
  For clauses $\cla \subseteq \clb$ of width at most~$w$ the  
  inequality
  \begin{equation*}
    \MultPoly{\cla} - \MultPoly{\clb} \geq 0
  \end{equation*} 
  has a derivation in \SAR of rank $w+1$ and size $O(w^2)$.
\end{lemma}

\begin{proof}
  Let $Y$ and $Z$ be the set of variables that occur positively
  and negatively, respectively, 
  \mbox{in $\clb \setminus \cla$},
  so that
  $\MultPoly{\clb}=\MultPoly{\cla}\cdot
  \MultPoly{Y,Z}$. 
  Note that 
  $\MultPoly{Y,Z}$
  is the product of the literals 
  \mbox{in $\clb \setminus \cla$},
  which are all variables in \SAR. For ease of notation, 
  let us write this product as
  $\prod^{|Y|+|Z|}_{i=1}v_{i}$.
  Then by using telescoping sums we can derive
  \begin{equation}
    \sum^{|Y|+|Z|}_{i=1} 
    (1-v_{i})
    \MultPoly{\cla} 
    \prod^{i-1}_{j=1} v_{i}
    = (1-\MultPoly{Y,Z})\MultPoly{\cla}
    = \MultPoly{\cla}-\MultPoly{\clb}
  \end{equation}
  which establishes the lemma.
\end{proof}

\begin{lemma}[Simulation of resolution step]
  \label{lem:SARinference}
  Let   $\cla$ and $\clb$ be clauses in which
  the variable~$x$ does not appear and let $w$ be the width of $\cla
  \vee \clb$. 
  Then the  inequality 
  \begin{equation*}
    \MultPoly{\cla \vee \varx} + \MultPoly{\clb \vee \olnot{\varx}} -
    \MultPoly{\cla \vee \clb} \geq 0 
  \end{equation*} 
  has a derivation in \SAR of rank $w+1$ and size $O(w^2)$.
\end{lemma}

\begin{proof}
  Using \reflem{lem:SARweakening} twice we derive the two inequalities
  $\MultPoly{\cla \vee \varx} - \MultPoly{\cla \vee \clb \vee \varx} \geq 0$ 
  and 
  $\MultPoly{\cla \vee \olnot{\varx}} - 
  \MultPoly{\cla \vee \clb \vee \olnot{\varx}} \geq 0$.
  Then we derive 
  $(\varx + \olnot{\varx} - 1)\MultPoly{\cla \vee \clb} \geq 0$ 
  from the axiom $\varx + \olnot{\varx} - 1 \geq 0$. This is the
  same as
  \begin{equation}
    \MultPoly{\cla \vee \clb \vee \olnot{\varx}} + \MultPoly{\cla \vee
      \clb \vee \varx} - \MultPoly{\cla \vee \clb} \geq 0\eqperiod
  \end{equation}
  The inequality that we want to prove is the sum of these three
  inequalities just derived.
  This \SAR derivation has size $O(w^{2})$ and rank $w$.
\end{proof}

\begin{remark}\label{stm:sa-rank-size}
  In Lemmas~\ref{lem:SARdownload},~\ref{lem:SARweakening}
  and~\ref{lem:SARinference} we gave the \SAR simulations of the
  steps of a resolution refutation. 
  To get a simulation in SA it is
  sufficient to substitute 
  $(1-x_{1}), \ldots, (1-x_{n})$
  for  the variables $\olnot{x}_{1}, \ldots,
  \olnot{x}_{n}$.
  %
  After the substitution we obtain a valid SA proof of the
  corresponding inequalities of the same rank, but potentially of
  larger size.
  Notice that the proofs of the inequalites in
  Lemmas~\ref{lem:SARdownload},~\ref{lem:SARweakening}
  and~\ref{lem:SARinference} have the form of 
  Equation~(\ref{eqn:SAdefinition}), with $O(w)$ axioms,  
  each of them multiplied by a degree $w+O(1)$ polynomial.
  Hence the size of each of these proofs is at most $O(w2^{w})$.
\end{remark}

Now we can show how resolution refutations can be efficiently
simulated in the SA and \SAR proof systems.

\begin{lemma}\label{lmm:SARsimulation}
  If a CNF formula $F$ has a resolution refutation of
  width $w$ and length $\lengthstd$, 
  then it has an SA refutation of rank
  $w+1$ and size $\Bigoh{w2^w \lengthstd}$ and an \SAR refutation of
  rank $w+1$ and size $\Bigoh{w^2 \lengthstd}$.
%
\end{lemma}

\begin{proof}
  Let $\proofstd=(\clc_{1}, \clc_{2}, \ldots, \clc_{\lengthstd})$ be a
  resolution refutation of $F$ where all clauses have width at most~$w$.
  Let us focus first on the \SAR simulation.
  For each clause~$\clc_i$ in the refutation we derive an inequality
  as follows:
  \begin{enumerate} \itemsep=0pt
    \item If $\clc_{i}$ is an axiom clause, then we derive
      $-\MultPoly{\clc_{i}} \geq 0$.

    \item
      If $\clc_{i}$ is obtained by weakening from $\clc_{j}$, 
      then we derive 
      $\MultPoly{\clc_{j}} -\MultPoly{\clc_{i}} \geq 0$.

    \item
      If $\clc_{i}$ is obtained by resolving $\clc_{j}$ and
      $\clc_{k}$, then we derive
      $
      \MultPoly{\clc_{j}} +\MultPoly{\clc_{k}} - \MultPoly{\clc_{i}}
      \geq 0$.
  \end{enumerate}
  All of these inequalities
  have \SAR derivations of rank $w+1$ and size $\Bigoh{w^{2}}$ by
  Lemmas~\ref{lem:SARdownload},~\ref{lem:SARweakening}
  and~\ref{lem:SARinference}
  (where we recall that the encoding of an axiom
  clause $\clc$ in the \SAR proof  system is $\SumPoly{\clc}-1 \geq
  0$, as required by \reflem{lem:SARdownload}).
  
  Now we have a sequence of inequalities 
  $Q_{1}\geq 0, Q_{2}\geq 0,
  \ldots, Q_{\lengthstd}\geq 0$, 
  where the inequality \mbox{$Q_{i} \geq 0$}
  corresponds to the clause~$\clc_{i}$ as explained above.
  Observe that any positive combination $\sum^{\lengthstd}_{i=1}
  \alpha_{i}Q_{i} \geq 0$ has a \SAR derivation of rank $w+1$ and size
  $O(\lengthstd \cdot k^{2})$.  
  In order to conclude the proof of the lemma, we just need to argue
  that there are positive weights $\alpha_{i}$ such that
  $\sum_{i}\alpha_{i}Q_{i} = -1$.

  The intuition is that 
  if $\clc_{i}$ is obtained by weakening from $\clc_{j}$ 
  then adding
  $Q_i =
  \MultPoly{\clc_{j}} -\MultPoly{\clc_{i}}$
  will cancel the term
  $-\MultPoly{\clc_{j}}$ in~$Q_j$
  representing~$\clc_j$,
  and if
  $\clc_{i}$ is inferred by resolution from $\clc_{j}$ and~$\clc_{k}$,
  then adding 
  $Q_i = \MultPoly{\clc_{j}} +\MultPoly{\clc_{k}} - \MultPoly{\clc_{i}}$ 
  will cancel the terms
  $-\MultPoly{\clc_{j}}$ and  $-\MultPoly{\clc_{k}}$
  representing
  $\clc_{j}$ and~$\clc_{k}$
  in $Q_j$ and $Q_k$, respectively.
  In the end, all monomials representing clauses are cancelled and the
  only term remaining is~$-1$.
  However, if a clause is used in several different applications of
  the resolution or weakening rules we need to set the weights so that
  it is cancelled the correct number of times.

  To do so,
  consider the DAG of the resolution refutation oriented from the
  initial clauses towards the empty clause. We assign a weight to each
  clause $\clc_{i}$ in this DAG inductively:  the empty clause
  $\clc_{\lengthstd}$ gets weight $1$, and if all immediate successors
  of a clause have already been assigned weights, then the clause gets
  the sum of the 
  weights of its immediate successors as the weight for itself.
  The value of $\alpha_{i}$ is then the weight assigned to the
  clause~$\clc_{i}$ in this way. 
  To verify that $\sum_{i}\alpha_{i}Q_{i} = -1$, notice that every
  polynomial $\MultPoly{\clc_{i}}$ has negative coefficient in the
  inequality $Q_{i}\geq 0$ and positive one in every $Q_{j}\geq 0$
  where $\clc_{i}$ appears as a premise in the derivation of
  $\clc_{j}$.
  By construction the coefficient of each $\MultPoly{\clc_{i}}$ in the
  final   sum is zero unless $i=\lengthstd$. Since
  $\alpha_{\lengthstd}=1$, the final sum is 
  equal to $-\MultPoly{\emptyset,\emptyset} $ which is~$-1$.

  We can obtain a simulation in the SA  proof system instead by
  substituting  $(1-x_{i})$ for every negative variable 
  $\olnot{x}$ in the \SAR simulation described above.
  Then we can reason as in \refrem{stm:sa-rank-size} to see that the
  the size and rank bounds claimed for SA hold. The lemma follows.
\end{proof}

\subsection{Lower bound on rank for Sherali-Adams resolution}
\label{sec:lower-bound-rank}

%
%
%

The \emph{pigeon-rank} of 
a Sherali-Adams resolution
refutation of $\TPHP$ 
of the form described in
Equation~\eqref{eqn:SAdefinition}  
is the maximum pigeon-degree of the polynomials 
to which the formulas
$\prod_{i \in
  \pLit_t} x_i \cdot \prod_{j \in \nLit_t} (1-x_j) \cdot \polyp_t$
expand.

In order to prove a lower bound on pigeon-rank it is useful to
generalize 
this concept
to a more abstract notion of rank for \SA proofs.
Let $V$ be a set of variables and let $H$ be a downward-closed family
of subsets of~$V$, i.e., such that if $Y$ belongs to $H$ and $X
\subseteq Y$, then $X$ also belongs to~$H$.
%
We say that a polynomial (or polynomial inequality) is
\introduceterm{\hbounded{}},
or has
\introduceterm{\hrank{}},
if $H$ contains the variable set of every
monomial 
in it.
We say that an \SA~derivation as in~\eqref{eqn:SAdefinition} has
\hrank if the polynomial to which each formula
$\prod_{i \in \pLit_t} x_i \cdot \prod_{j \in \nLit_t} (1-x_j) \cdot
\polyp_t$ expands is \hbounded.
Observe that if an \SA derivation has rank $r$, then it has \hrank
where $H$ is the family of all subsets of at most $r$
variables. Similarly, if an \SA refutation of $\TPHP$ has
pigeon-rank~$r$, then it has \hrank where $H$ is the family of all
subsets of variables that mention at most $r$~pigeons.

Let $\polyset{P}$ be a set of polynomial inequalities over the variable
set~$V$. We say that $\polyset{P}$ admits an
\introduceterm{\hconsistent family of distributions} 
if there exists a collection of probability distributions
$\set{\hdistr{X}}_{X \in H}$ 
over assignments $\set{0,1}^X$ 
as $X$ ranges over~$H$ that satisfy the following properties:
%
%
\begin{enumerate}[label=H\arabic*.,ref=H\arabic*]\itemsep=0pt
\item 
  \label{item:h-consistent-i}
  For every variable set $X \in H$ and every polynomial inequality
  $\polyq \geq 0$ in $\polyset{P}$ that has all its variables in~$X$,
  it holds that all assignments in the support of $\hdistr{X}$ satisfy
  $\polyq \geq 0$. 

\item 
  \label{item:h-consistent-ii}
  For every pair of variable sets $X,Y \in H $ such that $X \subseteq Y$ and
  for every assignment $\smallassmntSAR \in \{0,1\}^X$ it holds that
  \begin{equation}
    \label{eq:h-consistent-ii}
    \hdistr{X}(\smallassmntSAR) 
    = \sum_{\substack{\largeassmntSAR\in\{0,1\}^Y \\
        \largeassmntSAR \supseteq  \smallassmntSAR}} 
    \hdistr{Y}(\largeassmntSAR)
  \eqcomma
  \end{equation}
  where $\largeassmntSAR$ ranges over all assigments to $Y$ that are
  consistent with~$\smallassmntSAR$. 
\end{enumerate}
In the definition above and elsewhere, $\hdistr{X}(\smallassmntSAR)$
denotes the probability assigned to $\smallassmntSAR$ by the
distribution $\hdistr{X}$.  We will use such \hconsistent families of
distributions to establish the Sherali-Adams rank lower bound that we
need.  Before stating the formal lemma that we will appeal to, let us
try to provide some intuition.

%
%

If the set of polynomial inequalities $\polyset{P}$ were satisfiable
it would not be hard to come up with a family of probability
distributions with properties~\ref{item:h-consistent-i}
and~\ref{item:h-consistent-ii}: we could just fix a global probability
distribution over all 
satisfying assignments, and then let
$\hdistr{X}$ be the corresponding marginal distribution on any set of
variables~$X$.  For an unsatisfiable set $\polyset{P}$ there is no
such globally consistent family, but if we can find an \hconsistent
family of distributions for $\polyset{P}$, then~$\polyset{P}$ will
still ``look satisfiable'' to any derivation that does not go
``outside of~$H$.'' Whenever we look at a specific inequality $\polyq
\geq 0$ in~$\polyset{P}$, property~\ref{item:h-consistent-i} yields a
``marginal distribution'' that satisfies the inequality. Furthermore,
property~\ref{item:h-consistent-ii} ensures that such ``marginal
distributions'' over different sets look locally consistent.  The
following lemma makes this precise.

\begin{lemma} \label{lem:sufficient}
  Let $H$ be a downward-closed family of sets of variables and 
  let $\polyset{P}$ be a set of \hbounded polynomial inequalities.
  If $\polyset{P}$ has an \SA refutation of
  \hrank, then $\polyset{P}$ does not admit an \hconsistent family of
  distributions.
\end{lemma}

\begin{proof}
  Let us 
  think of each $X$ in $H$ as a new formal variable. For each monomial
  $\monom$, let $X_{\monom}$ denote the set of variables in~$\monom$. 
  If $\polyr$ is an \hbounded polynomial, let us write
  $\linform{\polyr}$ to denote the linear form on the variables $H$
  obtained from~$\polyr$ by replacing each term $c\cdot\monom$ by
  $c\cdot X_{\monom}$ and collecting all terms of the same variable
  into a single term by adding their coefficients (which could result
  in cancellations of terms).
%
%
%
%
Note that 
$\linform{\polyr}$ 
can also be thought of as 
the multilinearization of $\polyr$,
namely the polynomial obtained
from~$\polyr$ by removing all higher powers in the  monomials
to get $\linform{\monom} = X_{\monom}$ instead of~$\monom$.
We write $1_Y$ to denote the assignment $\set{ x \mapsto 1 : x \in Y
}$ to a set of variables~$Y$, and for a monomial~$\monom$ (multilinear
or 
not) we define $1_{\monom} = 1_{X_{\monom}}$.    

%
%

  Let   
  $\polyset{P} = \set{\polyq_1 \geq 0, \ldots, \polyq_m \geq 0 }$
  be a set of polynomial inequalities and suppose that there exists an
  \SA refutation of $\polyset{P}$ of the form~\eqref{eqn:SAdefinition} 
  that has \hrank. Let us write $\polyr_t$ for the polynomial to
  which the formula $\prod_{i \in \pLit_t} x_i \cdot \prod_{j \in \nLit_t}(1-x_j) \cdot
    \polyp_t$ expands
  for $1\leq t  \leq \tau$.
The assumption that the
refutation has \hrank means that every monomial in
the polynomial~$\polyr_t$ is  \hbounded. 

%
%

Assume for contradiction that $\polyset{P}$ admits an \hconsistent family
$\set{\hdistr{X}}_{X \in H}$. 
Let 
$\funcdescr{a}{H}{\reals}$ 
be the real-valued assignment defined by
\begin{equation}
  \label{eq:def-a}
  a(X) = \hdistr{X}(1_X)    
  \eqcomma
\end{equation}
\ie the probability of the all-ones assignment
to the  variables in~$X$ 
according to the distribution~$\hdistr{X}$,
and extend $a$ to all linear forms on the variables $X \in H$ linearly;
i.e., if $L = \sum_i c_i X_i$ is such a linear form with
coefficients $c_i$ and variables $X_i$, then 
$a(L) = \sum_i c_i \cdot a(X_i)$.
%

%
%
%
  We claim that $a$ 
  satisfies
  $a(\linform{\polyr}_t) \geq 0$ 
  for every 
  $1 \leq t \leq \tau$.
  By linearity it then further follows that
  $
  a\bigl(\sum_{t=1}^\tau \alpha_t \linform{\polyr}_t\bigr) 
  =
  \sum_{t=1}^\tau \alpha_t \cdot a \bigl( \linform{\polyr}_t\bigr) 
  \geq 0
  $, 
  which is a
  contradiction since $\sum_{t=1}^\tau \alpha_t \polyr_t = -1$ and hence
  also $\sum_{t=1}^\tau \alpha_t \linform{\polyr}_t = -1$.

%
%

  Let us prove that the assignment $a$ as defined in~\refeq{eq:def-a}
  satisfies every inequality
  $\linform{\polyr}_t \geq 0$
  for
  $1 \leq t \leq \tau$. We do so by establishing a stronger claim:
  if $X_t$ is the set of
  variables in $\polyr_t$ and $\YAExpectationNotation[X_{t}]$ denotes
  expectation under the   distribution $\hdistr{X_t}$, then 
  the following holds:
  \begin{enumerate}[label=A\arabic*.,ref=A\arabic*]\itemsep=0pt
  \item
      \label{item:a-property-i}
      The assignment
      $\funcdescr{a}{H}{\reals}$ 
      satisifies
      $a(\linform{\polyr}_t) = \YAExpectation[X_t]{\polyr_{t}}$.
    \item
      \label{item:a-property-ii}
      Every assignment in the support of $\hdistr{X_t}$ satisfies the
      inequality $\polyr_t \geq 0$.
  \end{enumerate}
  To see that 
  \ref{item:a-property-i}
  holds, we evaluate each monomial $\monom$ in $\polyr_t$ separately
  to get
  \begin{equation}
    \label{eq:multilinear-evalutation}
    a(\linform{{\monom}}) 
    = 
    a(X_{\monom}) 
    = \Pi_{X_{\monom}}(1_{{\monom}}) 
    = \sum_{\substack{\largeassmntSAR \in \set{0,1}^{X_t} \\
        \largeassmntSAR \supseteq 1_{{\monom}}}} 
    \hdistr{X_t}(\largeassmntSAR)     
    = \YAExpectation[X_t]{\monom}
    \eqperiod
  \end{equation}
  The first and second equalities
  in~(\ref{eq:multilinear-evalutation}) hold by definition; 
  the  third one follows from 
  property~\ref{item:h-consistent-ii}  of \hconsistent families of
  distributions;  and the
  final equality is true since a monomial~$\monom$ evaluates to~$1$
  under an assignment
  $\largeassmntSAR \in\set{0,1}^{X_t}$
  if and only if 
  $\largeassmntSAR$ 
  is compatible with~$ 1_{\monom} $. 
  Adding over all 
  terms
  we get
  $a(\linform{\polyr}_t) = \YAExpectation[X_t]{\polyr_t}$ by applying
  linearity of $a$ on the left and linearity of expectation on the right. 

  The verification of the claim in
  \ref{item:a-property-ii} is straightforward. Let 
  $\largeassmntSAR$ be an assignment in the support of $\hdistr{X_t}$. 
  Substituting the
  values assigned by $\largeassmntSAR$ to the variables of~$\polyr_t$,
  we deduce that
  \begin{equation}
    \label{eq:Pi-eval}
    \largeassmntSAR(\polyr_t) 
    =
    \largeassmntSAR \left( 
    \prod_{i \in \pLit_t} x_i \cdot \prod_{j \in \nLit_t}(1-x_j)
    \cdot \polyp_t
    \right)
    = 
    \prod_{i \in \pLit_t} \largeassmntSAR ( x_i ) \cdot
    \prod_{j \in \nLit_t} (1-\largeassmntSAR ( x_j )) \cdot
    \largeassmntSAR ( \polyp_t )
    \geq 0
    \eqperiod
  \end{equation}
  To see this, it suffices to 
  observe
  that all factors in the final
  expression in \refeq{eq:Pi-eval} are non-negative.  First,
  regardless of what the assignment $\largeassmntSAR$ is, we clearly
  have $ 0 \leq \largeassmntSAR ( x ) \leq 1$ for any variable $x$ in
  its domain and hence $\largeassmntSAR ( x_i ) \geq 0$ and
  $1-\largeassmntSAR ( x_j ) \geq 0$. Second, from
  property~\ref{item:h-consistent-i} we know that if $\polyp_t$ is one
  of the polynomials $\polyq_i$ in $\polyset{P}$ then $\largeassmntSAR
  ( \polyp_t ) \geq 0$ since $\largeassmntSAR$ is in the support of
  $\hdistr{X_t}$. And third, if $\polyp_t$ is one of the axioms $x_i^2
  - x_i$ or $x_i - x_i^2$ then $\largeassmntSAR ( \polyp_t ) = 0$
  since the range of $\largeassmntSAR$ is $\{0,1\}$, and if $\polyp_t$
  is the axiom $1$ then of course $\largeassmntSAR( \polyp_t ) = 1
  \geq 0$.  This concludes the proof of the lemma.
\end{proof}

Dantchev \etal~\cite{DantchevMartinRhodes2009Tight}
proved a rank lower bound on \SAR refutations of $\PHPk$. 
Let us show how this result can be extended to a 
pigeon-rank lower bound for~$\TPHP$.

\begin{lemma}
  \label{lmm:sar-pigeon-rank}
  Every \SAR refutation of $\TPHP$ has pigeon-rank at least $k$.
\end{lemma}

\begin{proof}
  First note that by replacing each variable $\dvarx$ by $1-\varx$ we
  transform an \SAR~proof into an \SA~proof of the same
  pigeon-rank. Thus, by Lemma~\ref{lem:sufficient} it will suffice to
  build an \hconsistent family of distributions where $H$ is the
  family of sets of variables that mention up to $k-1$ pigeons. 

  Intuitively, it is clear what the distributions should be: since
  there is room for up to \mbox{$k-1$~pigeons} in the pigeonholes, we can
  just choose any one-to-one mapping uniformly at random and set the
  Boolean variables accordingly. Formally, for every set $X$ of
  variables than mention at most $k-1$ pigeons we define the 
  distribution $\hdistr{X}$ as follows:
  \begin{enumerate} \itemsep=0pt
      \item 
        Let $A$ be the set of at most $k-1$ pigeons that are mentioned 
        by the variables in $X$.
      \item 
        Let $\smallmapSAR$ be a uniformly chosen one-to-one map
        $\smallmapSAR : A \rightarrow [k-1]$. 
  \item
    For $q_{v,w}\in X$ set $ q_{v,w}=1 $ if $\smallmapSAR(v) = w$,
    and $q_{v,w}=0$ otherwise. 
  \item 
    For $z_{v,w}\in X$ set $ z_{v,w}=1 $ if $\smallmapSAR(v) > w$,
    and $z_{v,w}=0$ otherwise. 
  \end{enumerate}
%
%

%
  Let us verify that a family of distributions 
  defined in this way
  satisfy  properties~\ref{item:h-consistent-i}  
  and~\ref{item:h-consistent-ii}.

  That property~\ref{item:h-consistent-i} is satisfied is immediate by
  construction. If $C$ is a clause in~$\TPHP$ with all variables
  contained in~$X$, then all assignments in the support of~$\hdistr{X}$
  satisfy~$C$ since they encode one-to-one mappings (with the
  extension variables~$z_{v,w}$ set appropriately).

  Property~\ref{item:h-consistent-ii} is also straightforward to verify.
  Fix any sets $X$ and $Y$ such that  $X \subseteq Y$ and that
  mention up to $k-1$ pigeons and any assignment
  $\smallassmntSAR \in \set{0,1}^X$. 
  Let $A$ and $B$ be the sets of at most $k-1$ pigeons that are mentioned 
  in $X$ and $Y$, respectively, and note that $A \subseteq B$.
  Let us write 
  $a=\setsize{A}$ 
  and 
  $b=\setsize{B}$. 
  By construction, the assignments $\largeassmntSAR \in \{0,1\}^Y$ in
  the support of~$\hdistr{Y}$ are in 
  bijective correspondence with the one-to-one mappings 
  $\funcdescr{\largemapSAR}{B}{[k-1]}$ 
  and the same holds for 
  $\smallassmntSAR$ in the support of $\hdistr{X}$ 
  \mbox{vis-a-vis} 
  $\funcdescr{\smallmapSAR}{A}{[k-1]}$. 
  Moreover, 
  each one-to-one mapping 
  $\funcdescr{\smallmapSAR}{A}{[k-1]}$ 
  can be chosen in
  $
  (k - 1)
  (k - 2)
  \cdots
  (k - a) = \binom{k-1}{a}a!
  $
  ways, and for a fixed~$\smallmapSAR$   
  the number of one-to-one mappings 
  $\funcdescr{\largemapSAR}{B}{[k-1]}$ 
  that extend~$\smallmapSAR$ is   
%
  $
  (k - a - 1)
  (k - a - 2)
  \cdots
  (k - b) = \binom{k-1-a}{b-a}(b-a)!
  $.
  Since all involved distributions are uniform over their support,  for
  $\smallassmntSAR \in \set{0,1}^X$ 
  in the support of~$\hdistr{X}$ we have
      \begin{equation}\label{eq:verifying-H-ii}
      \sum_{\substack{\largeassmntSAR \in \{0,1\}^Y \\ \largeassmntSAR \supseteq \smallassmntSAR}}
      \hdistr{Y}(\largeassmntSAR) = \sum_{\substack{\largeassmntSAR \, : \, \Pi_Y(\largeassmntSAR) > 0 \\ \largeassmntSAR \supseteq \smallassmntSAR}}
      \frac{1}{\binom{k-1}{b} {b!}} = \frac{\binom{k-1-a}{b-a} 
        (b-a)!}{\binom{k-1}{b} {b!}} = \frac{1}{\binom{k-1}{a} {a!}} = 
      \hdistr{X}(\smallassmntSAR) 
      \end{equation}
      and for $\smallassmntSAR$ outside the support of~$\hdistr{X}$ the
      whole summation in \eqref{eq:verifying-H-ii} is zero.  
      This finishes the proof of the lemma.
\end{proof}

\subsection{Size bounds for \PCR and \SAR refutations}
\label{sec:size-lower-bound}



Given the lower bounds on pigeon-degree and pigeon-rank for
refuting~$\TPHP$ in
\reftwolems{lmm:pcr-pigeon-degree}{lmm:sar-pigeon-rank}, respectively,
the size lower bounds on refutations of~$\erphpknk$ 
in polynomial calculus resolution and Sherali-Adams resolution are
straightforward adaptions of the lower bound for resolution in
\refth{thm:largenarrow}. We write down the details here for
completeness, starting with the \PCR bounds.

\begin{theorem}\label{thm:pcr-size-lowerbound}
  Let $k = k(n)$ be any integer-valued function such that 
  $k(n) \leq n/4\log n$. 
  Then  $\erphpknk$ can be refuted in \PCR in size
  $\Bigoh{k^{k+1} n^{k}}$,
  and any \PCR refutation requires size
  $\Bigomega{ \InverseRestrictionBoundInline{\lceil(k-1)/2\rceil} } $.
\end{theorem}

\begin{proof}
  %
  %
  Fix any \PCR refutation of   $\erphpknk$ and
  let $\monoset{M}$ be the set of monomials appearing in it.
  We hit the refutation with a random restriction~$ \rstd $
  distributed according to~$\PADistribution $.  
  Since restrictions preserve \PCR derivations
  we obtain a
  refutation of $ \restrict{\erphpknk}{\rstd} $, which as before is
  \TPHP
  after renaming of variables.

  Assume that $ \cardinality{\monoset{M}} <
  \InverseRestrictionBoundInline{\lceil(k-1)/2\rceil}$.
  Applying
  \reflem{lmm:StrongRestriction} with 
  $\ell = \left\lceil \frac{k-1}{2}\right\rceil$ 
  and taking a union bound over 
  the monomials in~%
  $ \monoset{M}$,
  we conclude that
  there must be at least one restriction~$ \rstd $ in the support of
  $\PADistribution $ such that the pigeon-degree of
  $\restrict{\proofstd}{\rstd}$ is at most $\left\lceil \frac{k-1}{2}
  \right\rceil -1$
  if $ n $ is large enough.
  This contradicts \reflem{lmm:pcr-pigeon-degree}, and hence
  $\cardinality{\monoset{M}}$ 
  must be at least 
  $\InverseRestrictionBoundInline{\lceil(k-1)/2\rceil}$.

  To obtain the upper bound we start with the resolution refutation in
  Theorem~\ref{thm:largenarrow}. It is not hard to see that any
  resolution refutation of size $S$ and width $w$ translates
  into a \PCR refutation of size $wS$ and degree $w+1$. The
  additional factor $w$ in the size is due to the fact that while
  resolution can arbitrary weaken a clause in one step, 
  the way multiplication is defined in \PCR  means that we need one
  multiplication step per literal to simulate the same weakening.
\end{proof}

The proof of the bounds for Sherali-Adams is very similar.

%
%

\begin{theorem}\label{thm:sar-size-lowerbound}
  Let $k = k(n)$ be any integer-valued function such that $k(n) \leq
  n/4\log n$. 
  Then
  $\erphpknk$
  can be refuted in \SAR in size
  $\Bigoh{k^{k+2} n^{k}}$,
  and any \SAR refutation requires size
  $\Omega \bigl( \InverseRestrictionBoundInline{k} \bigr) $.
\end{theorem}

 \begin{proof}
   %
   Fix any \SAR refutation of   $\erphpknk$ and 
   let $\monoset{M}$ be the set of monomials appearing in it.
   Hit the refutation with a random restriction
   $ \rstd $ distributed according to $ \PADistribution $.  
   Since restrictions preserve soundness of \SAR proofs, 
   this yields a refutation of $ \restrict{\erphpknk}{\rstd} $, which
   is \TPHP. 
 
 Suppose now that
   $ \cardinality{\monoset{M}} <\InverseRestrictionBoundInline{k}  $.
   Using \reflem{lmm:StrongRestriction} with $\ell = k$ and a union
   bound argument for~$\monoset{M}$, 
   we conclude that
   there 
   exists
   at least one restriction~$ \rstd $
   in the support of $\PADistribution $ such that the pigeon-rank of
   $\restrict{\proofstd}{\rstd}$ is at most $k-1$, 
   assuming 
   that
   $ n $ large enough.
   But
   this contradicts \reflem{lmm:sar-pigeon-rank},
   and hence the lower bound in the theorem follows.

   We obtain the upper bound by using the simulation in
   Lemma~\ref{lmm:SARsimulation} on the resolution refutation in
   Theorem~\ref{thm:largenarrow}.
 \end{proof}

%% file: lasserreupperbound.tex
\section{An upper bound for relativized PHP formulas in Lasserre}
\label{sec:upper-bound-lasserre}
\label{sec:lasserre}

%
%

In this section, we show that our lower bound 
\refth{th:main-thm-informal}
does not generalize to Lasserre
but that the formulas~$\erphpknk$ (and also~$\rphpknk$) have Lasserre
refutations in constant rank.
To establish this we we will use the easily verified identity
\begin{equation} 
  \label{eqn:identity}
  \sum_{\substack{i,j \in [n] \\ i \neq j}}
  \Big(1 - z_i - z_j\Big) z_j
  + (n-2)\sum_{j \in [n]} \Big(z_j^2 - z_j\Big) + \Big({1-\sum_{i \in [n]}
    z_i}\Big)^2 = 1 - \sum_{i \in [n]} z_i
\end{equation} 
a couple of times. A direct application 
of~\refeq{eqn:identity}
shows that the inequality 
$1 - \sum_{i \in [n]} z_i \geq 0$ 
has a \mbox{rank-$2$} Lasserre derivation from the set of all
inequalities of the form 
$1 - z_i - z_j \geq 0$ 
for 
$i,j \in [n]$, $i \neq j$.
We remark that this fact is a direct consequence of Lemma~1.5
in~\cite{LovaszSchrijver1991Cones}. 
Let us first use this
to get a \mbox{rank-$2$} Lasserre refutation of the standard
pigeonhole principle~$\PHPk$ encoded as the set of clauses
\begin{subequations}
  \begin{align}
    \label{eq:pigeon-ax}
    &
    x_{u,1} \lor x_{u,2} \lor \cdots \lor x_{u,k-1} 
    &&
    \text{$u \in [k]$,} 
    \\
    \label{eq:hole-ax}
    &
    \olnot{x}_{u,w} \lor \olnot{x}_{v,w}
    &&
    \text{$u,v \in [k]$, $u \neq v$, $w \in [k-1]$.}
  \end{align}
\end{subequations}
%
%
The proof we give next is essentially due to Grigoriev et al.~\cite{GHP02Complexity}.

\begin{lemma}[\cite{GHP02Complexity}]\label{lem:PHP}
  The formulas $\PHPk$ have Lasserre refutations of rank $2$.
\end{lemma}

\begin{proof}
  Combining all hole axioms 
  $1 - x_{u,w} - x_{v,w} \geq  0$ 
  in~\refeq{eq:hole-ax}
  for a fixed hole $w \in [k-1]$
  and  using~\eqref{eqn:identity}
  we can get the inequality
  $1-\sum_{u\in [k]} x_{u,w} \geq 0$.
  Adding these inequality over all holes $w \in [k-1]$ we  obtain
  \begin{equation}
    \label{eq:php-proof-i}
    k-1 - 
    \sum_{u \in [k]}
    \sum_{w \in [k-1]} 
    x_{u,w} \geq 0
    \eqperiod
  \end{equation}
  Adding together instead all the pigeon axioms
  $\sum_{w \in [k-1]} x_{u,w} - 1 \geq 0$ 
  in~\refeq{eq:pigeon-ax}
  we get
  \begin{equation}
    \label{eq:php-proof-ii}
    \sum_{u \in [k]} 
    \sum_{w \in [k-1]} 
    x_{u,w} - k \geq 0
    \eqperiod
  \end{equation}
  Summing
  \refeq{eq:php-proof-i}  and~\refeq{eq:php-proof-ii}
  yields
  $-1 \geq 0$.
%
\end{proof}

Two more 
applications
of~\eqref{eqn:identity} will help us get 
\mbox{rank-$9$} Lasserre refutations of $\rphpknk$ (and~$\erphpknk$)
by reduction to~$\PHPk$. The main idea of the proof is to substitute 
variables in the derivation
in Lemma~\ref{lem:PHP} with polynomials defined
over
the variables of~$\rphpknk$.

%
%

\begin{lemma}\label{lem:RPHP}
  The formulas
  $\rphpknk$ and $\erphpknk$ have Lasserre
  refutations of rank $9$.
\end{lemma}

\begin{proof}
  Let us first observe that we only need to present the Lasserre
  refutation of~$\rphpknk$. Once we have a refutation of the original
  formula $\rphpknk$ 
  we immediately obtain a refutation of the $3$\nobreakdash-CNF version
  $\erphpknk$  by  using the observation that
  the 
  encoding of a wide clause
  $\clc=\lita_{1} \lor \ldots \lita_{w}$ 
  is the the sum of the
  encodings of the  corresponding \mbox{$3$-clauses}
  $
  \lita_{1} \lor \lita_{2} \lor z_{2}, \,
  \mbox{$\olnot{z}_{2} \lor \lita_{3} \lor z_{3},$} \,
  \ldots, \,
  \olnot{z}_{w-2} \lor \lita_{w-1} \lor \lita_{w}
  $.
  This is so since all extension variables appear exactly once
  positively and exactly once negatively and so will simply cancel.
%
  Thus, 
  once we have a refutation of $\rphpknk$ we can get a valid
  refutation of $\erphpknk$ of the same rank by substituting the sum
  of the corresponding short axioms in $\erphpknk$ for any long axiom
  in~$\rphpknk$.

  For the rest of the proof we therefore focus on $\rphpknk$. Let 
  $\polyset{P}$ 
  be the set of polynomial inequalities that encode it and 
  let us define the shorthand
  \begin{equation}
    \label{eq:x-shorthand}
    x_{u,w}
    =
    \sum_{\ell \in [n]}
    p_{u,\ell} \, r_{\ell} \, q_{\ell,w}
    \eqperiod
  \end{equation}
  We  want to use the proof of the pigeonhole principle in
  \reflem{lem:PHP}
  together with the substitution~\refeq{eq:x-shorthand}  for~$x_{u,w}$. In
  order to do so, we
  need to show how to derive the substituted axioms used in that
  proof.  
  The inequalities $ x^{2}_{u,w} -x_{u,w} \geq 0$ 
  can be obtained by summing
  \begin{align}
    & 
    \sum_{\substack{\ell,m \in [n] \\ \ell \neq m}}
    (3 - q_{\ell,w} - q_{m,w} - r_{\ell} - r_{m}) 
    \, q_{\ell,w} \, q_{m,w} \, p_{u,\ell} \, p_{u,m}
    \, r_\ell \, r_m \geq 0 
    \eqcomma
    \\
    & 
    \sum_{\substack{\ell,m \in [n] \\ \ell \neq m}}
    (r^{2}_{\ell} - r_{\ell}) r_{m}  
    \, q_{\ell,w} \, q_{m,w} \, p_{u,\ell} \, p_{u,m}
    +
    (r^{2}_{m} - r_{m}) r_{\ell}  
    \, q_{\ell,w} \, q_{m,w} \, p_{u,\ell} \, p_{u,m}
    \geq 0 
    \eqcomma
    \\
    & 
    \sum_{\substack{\ell,m \in [n] \\ \ell \neq m}}
    q_{\ell,w}^{2} \, q_{m,w} \, p_{u,\ell} \, p_{u,m} \, r_\ell \, r_m 
    +
    q_{\ell,w} \, q_{m,w} ^{2}\, p_{u,\ell} \, p_{u,m} \, r_\ell \, r_m 
    \geq 0 
    \eqcomma
    \\
    \intertext{and}
    & \sum_{\ell \in [n]} 
    ( p_{u,\ell}^2 - p_{u,\ell} ) \, r_\ell^2 \, q_{\ell,w}^2 +
    ( r_\ell^2 - r_{\ell}) \, p_{u,\ell} \,  q_{\ell,w}^2 +
    ( q_{\ell,w}^2 - q_{\ell,w} ) \, p_{u,\ell} \, r_\ell \geq 0
    \eqcomma
  \end{align}
  and the latter inequalities all have direct \mbox{rank-$7$}
  derivations from~$\polyset{P}$. 
%
%
  To derive the inequalities $\sum_{w \in [k-1]} x_{u,w} - 1 \geq 0$
  for
  $u \in [k]$ we can sum up
  \begin{align}
    \label{eq:longineq2} 
    & 
    \sum_{\ell \in [n]} 
    \Big(\sum_{w \in [k-1]} q_{\ell,w} - r_\ell \Big) 
    p_{u,\ell} \, r_{\ell} \geq 0 
    \eqcomma
    \\ 
    & 
    \sum_{\ell \in [n]} 
    \Big(r^{2}_\ell-r_{\ell}\Big) p_{u,\ell} +
    \sum_{\ell \in [n]} \Big(r_\ell-p_{u,\ell}\Big) p_{u,\ell} +
    \sum_{\ell \in [n]} \Big(p^{2}_{u,\ell}-p_{u,\ell}\Big) \geq 0 
    \eqcomma
    \\
    \intertext{and}
    \label{eq:longineq}
    & 
    \sum_{\ell \in [n]} p_{u,\ell} - 1 \geq 0
    \eqcomma
  \end{align}
which can all be derived directly from~$\polyset{P}$ in rank~$3$.
%
%
%
%
The inequality $1 - x_{u,w} - x_{v,w}
\geq 0$ is the sum of
\begin{align}
  \label{eq:rqvariables1} 
  & 
  \sum_{\ell \in [n]} \Big(1 - p_{u,\ell} - p_{v,\ell}\Big)
  \, r_{\ell} \,   q_{\ell,w} \geq 0 
  \\
  \intertext{and}
  \label{eq:rqvariables2}
  &  1-\sum_{\ell \in [n]} r_{\ell} q_{\ell,w} \geq 0
  \eqcomma
\end{align}
where \eqref{eq:rqvariables1} has a direct rank-$3$ derivation
from~$\polyset{P}$.  For~\eqref{eq:rqvariables2} we need to do some
more work. Fix 
indices $\ell,m \in [n]$ with $\ell \neq m$ and observe that
%
%
\begin{multline}
  \bigl(
  1 - r_{\ell}q_{\ell,w} - r_{m}q_{m,w}
  \bigr)
  \, r_{\ell}q_{\ell,w} =  
  \\ 
  \bigl(
  3 - r_\ell - r_m - q_{\ell,w} - q_{m,w}
  \bigr) 
  \, q_{\ell,w} \, r_{m} \, q_{m,w} + 
  \bigl(
  q^{2}_{\ell,w} - q_{\ell,w}
  \bigr)
  \, r_{m} \, q_{m,w} 
  + 
  \bigl(
  q^{2}_{m,w} - q_{m,w}
  \bigr)
  \, r_{m} \, q_{\ell,w} 
  \\
  +  
  \bigl(
  r^{2}_{m} - r_{m}
  \bigr)
  \, q_{\ell,w} \, q_{m,w} +  
  \bigl(
  r_{\ell} - r^{2}_{\ell}
  \bigr)
  \, q_{\ell,w} + 
  \bigl(
  q_{\ell,w} - q^{2}_{\ell,w}
  \bigr)
  \, r^{2}_{\ell}
  \eqperiod
\end{multline}
%
Note that the first term 
on
the right-hand side of this equation is
the polynomial translation of axiom~\eqref{eq:QinjectiveWide}.
Writing 
$z_\ell$ for $r_{\ell} \, q_{\ell,w}$, this shows that the inequality
$(1-z_\ell-z_m) z_\ell \geq 0$ has a \mbox{rank-$4$} derivation from
$\polyset{P}$. Combined with the fact that $z_\ell^2 - z_\ell = (r^2_\ell -
r_\ell) q^2_\ell + (q^2_\ell - q_\ell) r_\ell$,
equation~\eqref{eqn:identity} gives a \mbox{rank-$4$} derivation of 
$1 - \sum_{\ell \in [n]} z_\ell \geq 0$. This is
precisely~\eqref{eq:rqvariables2}.

Now we mimic the refutation of $\PHPk$ in Lemma~\ref{lem:PHP}.  For
a fixed $w \in [k-1]$ we can use the derivations of 
$1 - x_{u,w} - x_{v,w} \geq 0$ 
and 
$x_{v,w}^2 - x_{v,w} \geq 0$ 
in combination with~\eqref{eqn:identity} 
to obtain the inequality 
$1 - \sum_{u \in [k]} x_{u,w} \geq 0$ 
by a rank-$9$ derivation. Adding all such inequalities for
$w \in [k-1]$ 
gives 
\begin{equation}
  k - 1 - 
  \sum_{w \in [k-1]} 
  \sum_{u \in [k]} 
  x_{u,w} \geq 0
\eqperiod  
\end{equation}
On the other hand, adding $\sum_{w \in [k-1]} x_{u,w} - 1 \geq 0$
over all $u \in [k]$
yields
\begin{equation}
  \sum_{u \in [k]} \sum_{w \in [k-1]}
  x_{u,w} - k \geq 0
\end{equation}
in rank~$3$
(the rank of the
derivation of $\sum_{w \in [k-1]} x_{u,w} - 1 \geq 0$), and a final
addition allows us to derive
$-1 \geq 0$, 
never going above rank~$9$.
%
%
%
\end{proof}

%% file: concludingremarks.tex
\section{Concluding remarks}
\label{sec:concluding-remarks}

In this paper, we exhibit
a family of $3$-CNF formulas over $n$~variables that
can be refuted in resolution in width~$w$ but require refutations 
of size $n^{\bigomega{w}}$. Furthermore, this lower bound can be extended to 
polynomial calculus resolution (PCR) and Sherali-Adams.
This shows that the seemingly naive counting upper bounds on proof
size in terms of width for resolution, degree for PCR, and rank
for Sherali-Adams are actually all tight up to small constant factors in the
exponent. Furthermore, 
our lower bound for resolution 
also implies that the result
in~\cite{AFT11ClauseLearning}  
that CNF formulas  refutable in width~$w$ can be decided by CDCL
solvers in time~$n^{\bigoh{w}}$ is 
optimal
(again up to constant factors
in the exponent), since any resolution refutation the solver finds 
might have to be that large in the worst case.

Regarding open problems, perhaps the most obvious one concerns the
tightness of our result. 
Our formulas have roughly
$N = n^2$ variables and are refutable in width roughly
$w = 2k$, and our size lower bound is on the order of 
$n^{k} = N^{w/4}$. 
However, the direct counting argument for width $w$
gives an upper bound of about $N^{w}$ clauses. Could this gap in the
exponent be closed?  
If so, this would have to be for a different formula family since
ours has an upper bound of
roughly
\mbox{$n^k = N^{w/4}$.}
One point worth noting is that one can
shave a factor~$2$ off the gap in the exponent by considering the
$4$\nobreakdash-CNF formulas obtained if the $4$\nobreakdash-clauses
in~\eqref{eq:QinjectiveWide} 
are \emph{not} converted to \mbox{$3$-CNF}.
In this case, the same upper and lower bounds still hold, but the
number of variables is on the order of~$N = kn$,
which means that we get a lower bound
of the form~$N^{w/2}$
if we focus on width~$w$ upper-bounded by a constant.

%

A more fundamental question is whether we can find a formula family
that exhibits the same kind of hardness for Lasserre. 
As shown in this paper, the formulas we used for resolution,
\PCR and Sherali-Adams will not work.
For tree-like Lov{\'a}sz-Schrijver (LS), however, we believe that our formulas
should be hard (and that the method of proof should be similar, with long
\emph{paths} in the refutation tree playing the role of long
monomials). In view of the Lasserre upper bound, for tree-like~$\LSPLUS$
we do not know what to believe.
The main problem with our formulas is that after restriction we obtain
a pigeonhole principle which is hard for resolution, \PCR, and Sherali-Adams
(in term of rank) but easy for LS$^{+}$.
A way to get a similar lower bound for Lasserre might be to find a formula
that is hard for Lasserre rank and 
becomes hard for Lasserre
size after relativization.

A natural formula for which it would be interesting to prove
similar size lower bounds as in this paper is the so-called
\introduceterm{clique formula} claiming that there is a 
\mbox{$k$-clique} in some fixed $n$-vertex graph chosen so that this
claim is false. It has been conjectured (e.g., in~%
\cite{BeyersdoffEtAl2012Parameterized})
that such formulas require resolution refutation size
$n^{\bigomega{k}}$ 
for the right kind of graphs, and this has been proven for
the restricted case of tree-like
resolution~\cite{BeyersdorffGalesiLauria2013parameterized}.
If such a lower bound could be established for general resolution, it
would have interesting consequences for parameterized proof complexity.

Finally, while the relations between size, width, and space in
resolution are now fairly well-understood, one big open question
remains. Namely, it was shown in~\cite{BW01ShortProofs}
that if a formula has a short resolution refutation then it can also
be refuted in small width, but this narrow refutation is obtained at
the price of an exponential blow-up in size. Is this inherent, or is
it just an artifact of the proof in~\cite{BW01ShortProofs}? 
That is, can size and width be
optimized simultaneously in resolution, or are there formulas for
which optimizing one of the measures must always cause a stiff penalty
for the other? 
For size vs.~space and space vs.~width dramatic trade-offs are known~%
\cite{BBI12TimeSpace,Ben-Sasson09SizeSpaceTradeoffs,BN11UnderstandingSpace},
and these results extend also to \PCR~\cite{BNT12SomeTradeoffs},
but it remains 
open whether there are similiar trade-offs between
size and width in resolution or between size and degree in \PCR.

%% file: acknowledgements.tex
\section*{Acknowledgments}

The authors would like to thank Mladen Mik\v{s}a and Marc~Vinyals for 
interesting discussions related to the topics of this work.

Part of the work of \theauthorAA was done while visiting KTH
Royal Institute of Technology.
The second and third authors were funded by the
European Research Council under the European Union's Seventh Framework
Programme \mbox{(FP7/2007--2013) /} ERC grant agreement no.~279611.
\TheauthorJN 
was also supported by
Swedish Research Council grants 
\mbox{621-2010-4797}
and
\mbox{621-2012-5645}.